\documentclass[11pt,letterpaper]{article}
\usepackage{amssymb}
\usepackage{graphics}
\usepackage[dvips]{epsfig}
\usepackage{fullpage}

\newtheorem{theorem}{Theorem}
\newtheorem{lemma}{Lemma}
\newtheorem{claim}{Claim}

\newtheorem{fact}{Fact}
\newtheorem{corollary}{Corollary}

\newtheorem{definition}{Definition}
\newtheorem{example}{Example}

\newtheorem{observation}{Observation}
\newcommand{\qed}{\hfill $\Box$ \bigbreak}
\newenvironment{proof}{\noindent{\bf Proof.~}}{\qed}

\def\cS{{\cal S}}
\def\cD{{\cal D}}
\def\cF{{\cal F}}

\def\cM{{\cal M}}

\newcommand{\apex}{\mbox{\sc apex}}

\def\cI{{\cal I}}
\def\cD{{\cal D}}
\def\cF{{\cal F}}
\def\cT{{\cal T}}
\def\cM{{\cal M}}

\def\cP{{\cal P}}

\newcommand{\dfs}{\mbox{\sc dfs}}

\newcommand{\weight}{\mbox{\rm weight}}
\newcommand{\parent}{\mbox{\rm parent}}

\newcommand{\SD}{\mbox{\sc Spine-Depth}}

\begin{document}

\title{An Optimal Ancestry Labeling Scheme \\ with Applications to XML Trees and  Universal Posets\thanks{Preliminary results of this paper appeared in the proceedings of the 42nd ACM Symposium on Theory of Computing (STOC), 2010, the 21st ACM-SIAM Symposium on Discrete Algorithms (SODA), 2010, and the 21st ACM Symposium on Parallel Algorithms and Architectures (SPAA), 2009, as part of 
\cite{FK-spaa09,FK09,FK10}. 
This research is supported in part by the ANR project DISPLEXITY, and by the INRIA
project GANG.}}

\author{
Pierre Fraigniaud\\[1ex]
{\small CNRS and Univ. Paris Diderot}\\{\small\sl pierre.fraigniaud@liafa.univ-paris-diderot.fr}
\and
Amos Korman\\[1ex]
{\small CNRS and Univ. Paris Diderot}\\{\small\sl amos.korman@liafa.univ-paris-diderot.fr}
}

\date{}

\maketitle

\begin{abstract}
In this paper we solve the  {\em ancestry}-labeling scheme problem which aims at assigning the shortest possible labels (bit strings) to nodes of rooted trees, so that ancestry queries between any two nodes can be answered  by inspecting their assigned labels only. This problem was  introduced more than twenty years ago by  Kannan et al. [STOC~'88], and is among the most well-studied problems in the field of informative labeling schemes.
We construct an ancestry-labeling scheme for $n$-node trees with
label size  $\log_2 n + O(\log\log n)$ bits, thus matching the
$\log_2 n + \Omega(\log\log n)$ bits lower bound given by Alstrup et al. [SODA~'03].
Our scheme is based on a simplified ancestry scheme that operates extremely well on a restricted set of trees. In particular, for the set of $n$-node trees with depth at most $d$, the simplified ancestry scheme enjoys label size of $\log_2 n+2\log_2 d +O(1)$ bits. Since the depth of most XML trees is at most some small constant, such an ancestry scheme may be of practical use. In addition, we also obtain an {\em adjacency}-labeling scheme 
 that labels $n$-node trees of depth  $d$ with labels of size
$\log_2 n+3\log_2 d +O(1)$ bits.
All our schemes assign the labels in linear time, and
guarantee that any  query can be answered in constant time.

Finally, our ancestry scheme finds applications to the construction of small {\em universal} partially ordered sets (posets).
Specifically, for any fixed integer~$k$, it enables the construction of a universal poset of size $\tilde{O}(n^k)$ for the family of $n$-element posets with tree-dimension at most~$k$.
Up to lower order terms, this bound is  tight thanks to a lower bound of $n^{k-o(1)}$ due to Alon and Scheinerman [Order~'88].

\end{abstract}

\thispagestyle{empty}
\newpage
\setcounter{page}{1}

\section{Introduction}
\label{section:Introduction}

\subsection{Background and motivation}

How to  represent a graph in a compact manner is a fundamental data structure question.
In most traditional graph representations, the names (or identifiers)
given to the nodes serve
merely as pointers to entries in a data structure, and thus reveal no information about the graph structure per se. Hence, in a sense, some memory space is  wasted for the storage of content-less data.
In contrast, Kannan et al. \cite{KNR92} introduced the notion of {\em informative labeling schemes},
which involves a mechanism  for assigning short, yet informative, labels to nodes.
Specifically, the goal of such schemes is to assign  labels to nodes in such a way that allows one to infer
information regarding any two nodes {\em directly from their labels}. 
As explicated below, one important question  in the
framework of informative labeling schemes is how to efficiently encode the ancestry relation in trees.
This  is formalized as follows.

\paragraph{The ancestry-labeling scheme problem:}
Given any $n$-node rooted tree $T$,
label the nodes of~$T$ using the shortest possible labels (bit strings) such that, given any pair of nodes $u$ and $v$ in $T$,
one can  determine whether $u$ is an {\em ancestor} of $v$ in $T$ by
merely inspecting the labels of $u$ and $v$.

\bigskip

The  following
simple ancestry-labeling scheme was suggested in \cite{KNR92}.
Given a rooted $n$-node tree $T$, perform a DFS traversal in $T$ starting at the root,
and provide each node $u$ with a
DFS number, $\dfs(u)$, in the range  $[1,n]$. (Recall, in a DFS traversal, a node is visited before any of its children, thus, the DFS number 
of a node is smaller than the DFS number of any of its descendants). The label of a node $u$ is simply the interval $I(u)= [\dfs(u),\dfs(u')]$, where
 $u'$ is the descendant of $u$ with the largest DFS number.
An ancestry query then amounts to an interval containment query between the corresponding labels: a node $u$ is an ancestor of a node $v$  if and only if $I(v)\subseteq I(u)$.
Clearly, the {\em label size}, namely, the maximal number of bits in a label assigned by this ancestry-labeling scheme to any node in any $n$-node tree, is bounded by $2\log n$ bits\footnote{All logarithms in this paper are taken
in base 2.}.

The $2\log n$ bits scheme of \cite{KNR92} initiated an extensive research \cite{AAKMT01,AKM01,ARSODA02,KM01,KMS02,TZ01}
whose goal was to reduce the label size of ancestry-labeling schemes as much as possible. The main motivation behind these works
lies in the fact that a small improvement in the label size of ancestry-labeling schemes may contribute to a significant improvement in the performances of XML search engines.
Indeed, to implement sophisticated queries,
XML documents  are viewed as labeled trees, and
typical queries over the documents amount to
testing relationships between document items, which correspond to ancestry queries
among the corresponding tree nodes \cite{ABS99,DFFLD99,xsl,xslt}.
XML search engines process such queries using an index structure that summarizes this ancestry information. To allow good performances,  a large portion of the XML index structure
resides in the main memory.
Hence, the length of the labels is a main factor which determines the index size. Thus, due to the enormous size
of the Web data, even a small reduction in the label size may contribute  significantly  to both
 memory cost reduction and performance improvement.
A  detailed explanation regarding this application  can be found in various papers on ancestry-labeling schemes (see, e.g., \cite{AAKMT01,KMS02}).

In \cite{ABR05}, Alstrup et al. proved a lower bound of $\log n + \Omega(\log\log n)$ bits
for the  label size of an ancestry-labeling scheme.  On the other hand, thanks to a scheme by Abiteboul et al. \cite{AAKMT01}, the current state of the art upper bound
is $\log n + O(\sqrt{\log n})$ bits. Thus,
  a large gap is still left between
the best known upper and  lower bounds
 on the label size.
The main result of this paper closes the gap. This is obtained by constructing an ancestry-labeling scheme whose label size matches the aforementioned lower bound.

Our scheme is based on a simplified ancestry scheme that operates extremely well on a restricted set of trees. In particular, for the set of $n$-node trees with depth at most $d$, the simplified ancestry scheme enjoys label size of $\log_2 n+2\log_2 d +O(1)$ bits.
This result can be of independent interest for  XML search engines, as a typical XML tree has extremely small depth (cf. \cite{CKM02,DG2006,MTP06,MBV05}).
For example, by examining  about 200,000
XML documents on the Web, Mignet et al.~\cite{MBV05}
found  that \emph{the average depth of an XML tree
is 4, and that 99\% of the trees have depth at most 8}.
Similarly, Denoyer and Gallinari \cite{DG2006} collected about 650,000 XML trees taken from the Wikipedia
collection\footnote{XML trees taken from the Wikipedia
collection have actually relatively larger depth compared to usual XML trees~\cite{Den09}.}, and found that \emph{the average depth of a node is 6.72}.
 
In addition, our 
 ancestry-labeling scheme on arbitrary trees finds applications in the context of  universal partially ordered sets (posets). Specifically, the bound on the label size translates 
 to an upper bound on the size of the smallest  universal poset  for the family
of all $n$-element posets with tree-dimension at most $k$ (see Section \ref{sec:preliminaries} for the definitions).
It is not difficult to show  that the smallest size of such a universal poset  is at most $n^{2k}$.
On the other hand, it follows from a result by Alon and Scheinerman \cite{AS88} that
  this size is also at least $n^{k-o(1)}$. As we show, it turns out that the real bound is much closer to this lower bound
than to the $n^{2k}$ upper bound.

\subsection{Related work}\label{sec:related}

\subsubsection{Labeling schemes}

As mentioned before, following the $2\log n$-bit ancestry-labeling scheme in \cite{KNR92}, a considerable amount of research has been devoted to improve the upper bound on the label size as much as possible. Specifically, \cite{AKM01} gave a first non-trivial upper bound of $\frac{3}{2}\log n + O(\log\log n)$ bits. In~\cite{KM01}, a scheme with label size $\log n + O(d\sqrt{\log n})$  bits was constructed to detect ancestry only
 between nodes at distance at most $d$ from each other. An ancestry-labeling scheme with  label size of $\log n +O(\log n/\log\log n)$ bits was given in \cite{TZ01}. The current state of the art upper bound of  $\log n + O(\sqrt{\log n})$  bits was given in \cite{ARSODA02} (that scheme
is described in detail in the journal publication~\cite{AAKMT01} joint with \cite{AKM01}).
Following  the aforementioned results on ancestry-labeling schemes for general rooted trees,  \cite{KMS02} gave 
an experimental comparison of different ancestry-labeling schemes  over XML tree instances that appear in ``real life''.


The ancestry relation is the transitive closure of the {\em parenthood} relation. Hence, the following parenthood-labeling scheme problem is inherently related to the ancestry-labeling scheme problem:
given a rooted tree $T$, label the nodes of $T$ in the most compact way such that
 one can  determine whether $u$ is a  {\em parent} of $v$ in $T$ by merely inspecting the corresponding labels.
 The parenthood-labeling scheme problem was also introduced in  \cite{KNR92},
and  a very simple parenthood
scheme was constructed there,  using labels of size
at most $2\log n$ bits. (Actually,  \cite{KNR92} considered  {\em adjacency}-labeling schemes in trees rather than parenthood-labeling schemes, however,
such schemes are equivalent up to a constant number of bits in the label size\footnote{To see this equivalence,
observe that one can
construct a parenthood-labeling  scheme from an adjacency-labeling  scheme in trees,
as follows.
Given a rooted tree $T$, first label the nodes of $T$ using the adjacency-labeling  scheme (which ignores the fact that $T$ is rooted). Then,
for each node $u$, in addition to the label given to it by the adjacency-labeling scheme, add two more bits,
 for encoding $d(u)$, the distance from $u$ to the root, calculated modulo 3. Now the parenthood-labeling scheme follows by observing that
for any two nodes $u$ and $v$ in a tree, $u$ is a parent of $v$ if and only if $u$ and $v$ are adjacent and $d(u)=d(v)-1$ modulo 3.}). By now, the parenthood-labeling scheme problem is almost completely closed thanks to Alstrup  and Rauhe \cite{AR02+}, who constructed a parenthood  scheme for $n$-node trees with label size  $\log n +O(\log^* n)$ bits. In particular, this bound indicates that encoding ancestry in trees is strictly
more costly than  encoding parenthood.

Adjacency labeling schemes where studied for other types of graphs, including, general graphs~\cite{Alstrup-STOC}, bounded degree graphs \cite{Noy-ICALP}, and planar graphs \cite{Gavoille-Planar}.
Informative labeling schemes were also proposed for other graph problems,
including distance \cite{ABR05,GPPR01,T01},
routing \cite{FG01,TZ01}, flow \cite{KKKP04,KK06}, vertex connectivity \cite{HL09,KKKP04,K07}, and nearest common ancestor in trees \cite{AGKR01,Alstrup-NCA,Peleg00:lca}.

Very recently, Dahlgaard et al. \cite{Ancestry-noy} and Alstrup et al. \cite{Alstrup-FOCS} claim to provide asymptotically optimal schemes for the ancestry problem and the adjacency problem on trees, respectively. 

\subsubsection{Universal posets}

When considering infinite posets, it is known that
a countable universal poset for the family of all countable posets exists.
This classical result was proved several times \cite{F53,J56,J57} and, in fact, as mentioned in \cite{HL05}, has
 motivated the whole research area of category theory.
 
We later give a simple relation between the label size of  {\em consistent} ancestry-labeling schemes and the size of 
universal posets for the family
of all $n$-element posets with {\em tree-dimension} at most~$k$ (see Section \ref{sec:preliminaries} for the corresponding definitions).
The $2\log n$-bit ancestry-labeling scheme of \cite{KNR92} is consistent,
 and thus it provides yet another evidence for the existence of a
 universal poset with~$n^{2k}$ elements for the family
of all $n$-element posets with tree-dimension at most $k$.
It is not clear whether the ancestry-labeling schemes in  \cite{AKM01,ARSODA02,KM01,TZ01}  can be somewhat  modified to be consistent
and still maintain the same label size. However, even if this is indeed the case, the   universal poset  for the family
of all $n$-element posets with tree-dimension at most $k$ that would be obtained from those schemes,
 would be of size $\Omega(n^k 2^{k\sqrt{\log n}})$.

The lower bound of
\cite{ABR05}
 implies a lower bound of
$\Omega(n\log n)$ for the number of elements in a universal poset for the family of
$n$-element posets with tree-dimension $1$.
As mentioned earlier, for fixed $k>1$, the result of Alon and Scheinerman \cite{AS88}
implies
 a lower bound of $n^{k-o(1)}$ for the number of elements in a universal poset for the family of
$n$-element posets with tree-dimension at most $k$.

\subsection{Our contributions}

The main result of this paper provides  an  ancestry-labeling scheme for $n$-node  rooted trees,
 whose label size is $\log n + O(\log\log n)$ bits. This scheme assigns the labels to the nodes of any tree in linear time and   guarantees that any ancestry query is answered in constant time.
By doing this, we solve the ancestry-labeling scheme problem which is among the main open problems in the field of informative labeling schemes. 

Our main scheme is based on a simplified ancestry scheme that is particularly efficient on a restricted set of trees, which includes the set of $n$-node trees with depth at most $d$. For such trees, the simplified ancestry scheme enjoys label size of $\log_2 n+2\log_2 d +O(1)$ bits. 
A simple trick allows us to use this latter ancestry-labeling scheme for designing a \emph{parenthood}-labeling scheme for $n$-node
 trees of depth at most $d$ using labels of
 size $\log n +3\log d +O(1)$ bits. Each of these two schemes assigns the labels to the nodes of any tree in linear time. The schemes also guarantee that the corresponding queries are answered in constant time.
 
Our schemes rely on two novel tree-decompositions. The first decomposition, called \emph{spine} decomposition, bears similarities with the classical heavy-path decomposition of Sleator and Tarjan~\cite{Sleator}. It is used for the construction of our simplified ancestry-labeling scheme. Our main ancestry-labeling scheme uses another tree-decomposition, called \emph{folding} decomposition. The spine decomposition of the folding decomposition of any tree has a crucial property, that is heavily exploited in the construction of our main labeling scheme. 

Finally, we  establish a simple relation between compact ancestry-labeling schemes and small universal posets. Specifically,
we show that  there exists a  {\em consistent} ancestry-labeling scheme for $n$-node forests  with label size $\ell$ 
 if and only if, for any integer $k\geq 1$, there exists a universal poset with
$2^{k\ell}$ elements for the family of $n$-element posets with tree-dimension at most $k$.
Using this equivalence, and slightly modifying our ancestry-labeling scheme,
we prove that for any integer~$k$, there exists a universal poset of size $\tilde{O}(n^k)$ for the family
of all $n$-element posets with tree-dimension at most $k$. Up to lower order terms\footnote{The $\tilde{O}$ notation hides polylogarithmic terms.}, this bound is tight.

\subsection{Outline}

Our paper is organized as follows. Section~\ref{sec:preliminaries} provides the essential  definitions, including the definition of the spine decomposition. 
In Section~\ref{sec:bounded} we describe our labeling schemes designed for a restricted family of trees, which includes trees of bounded depth.
The main result regarding the construction of the optimal ancestry-labeling scheme is presented
in Section \ref{sec:main}.
Our result concerning small universal posets appears in Section~\ref{sec:poset}.   Finally, in Section~\ref{sec:conclusion}, we 
conclude our work and introduce some directions for further research on randomized labeling schemes. 

\section{Preliminaries}\label{sec:preliminaries}

Let $T$ be a {\em rooted} tree, i.e., a tree with a designated node $r$ referred as the {\em root} of $T$. A rooted forest is a forest consisting of several rooted trees.
The {\em depth} of a node $u$ in some (rooted) tree $T$
is defined as the smallest number of nodes on the path leading from $u$ to the root.
In particular, the depth of the root is 1. The  depth of a rooted tree is defined as the maximum depth over all its nodes,
and the depth of a rooted forest is defined as the maximum depth over all the trees in the forest.

For two nodes $u$ and $v$ in a rooted tree $T$,
we say that $u$ is an {\em ancestor} of $v$ if $u$ is one of the nodes on the shortest path  in $T$ connecting $v$ and the root $r$.  (An ancestor of $v$  can be $v$ itself; Whenever we consider an ancestor $u$ of a node $v$, where $u\neq v$, we refer to $u$ as a \emph{strict} ancestor of~$v$). 
For two nodes $u$ and $v$ in some (rooted) forest $F$, we say that $u$ is an {\em ancestor} of $v$ in $F$ if and only if 
$u$ and $v$ belong to the same rooted tree in $F$, and $u$ is an ancestor of $v$ in that tree.
 A node $v$ is a {\em descendant} of $u$ if and only if $u$ is an ancestor of $v$. For every non-root node $u$, let $\parent(u)$ denote the parent of $u$, i.e., the ancestor of $u$ at distance~1 from it. 

 The {\em size} of $T$, denoted by $|T|$, is
the number of nodes in $T$.
The {\em weight} of a node $u\in T$, denoted by $\weight(u)$,
is defined as  the number of descendants of $u$, i.e., $\weight(u)$ is the size of the subtree hanging down from $u$.
In particular, the weight of the root is  $\weight(r)=|T|$.

For every integer $n$, let $\cT(n)$ denote the family of all rooted trees of size at most $n$, and  let $\cF(n)$ denote
the family of all forests of rooted trees, were each forest in  $\cF(n)$ has at most  $n$ nodes.


For two integers $a\leq b$, let $[a,b]$ denote the set of integers
$\{a,a+1,\cdots, b\}$. (For $a<b$, we sometimes use the notation $[a,b)$ which simply denotes the set of integers
$\{a,a+1,\cdots, b-1\}$). We refer to this set as an {\em interval}.
For two intervals $I=[a,b]$ and $I'=[a',b']$, we say that $I\prec I'$ if $b<a'$.
The {\em size} of an interval $I=[a,b]$ is $|I|=b-a+1$,
namely, the number of integers
in~$I$.

\subsection{The spine decomposition}\label{sub-spine}

Our ancestry scheme uses a novel decomposition of trees, termed  the {\em spine} decomposition (see Figure~\ref{fig:spine}). This decomposition bears similarities to the classical {\em heavy-path} decomposition of Sleator and Tarjan~\cite{Sleator}. Informally, the spine decomposition is based on a path called {\em spine} which starts from the root, and goes down the tree along the heavy-path until reaching a node whose heavy child has less than half the number of nodes in the tree. 
This is in contrast to the heavy-path which goes down until reaching a node $v$ whose heavy child has less than half the number of nodes in the subtree rooted at $v$. Note that the length of a spine is always at most the length of the heavy-path but can be considerably smaller. (For example, by augmenting a complete binary tree making it  slightly unbalanced, one can create a tree with heavy-path of length $\Omega(\log n)$ while its spine is of length $O(1)$.)

Formally, given a tree $T$ in some forest $F$, we define the {\em spine} of $T$ as the following path $S$. 
Assume that each node $v$ holds its weight  $\omega(v)$ (these weights can easily be computed in linear time). We define the construction of $S$ iteratively. In the $i$th step, assume that the path $S$ contains the vertices $v_1,v_2,\cdots v_i$, where $v_1$ is the root $r$ of $T$~and~$v_j$ is a child of~$v_{j-1}$, for $1<j\leq i$. If the weight of a child of $v_i$ is more than  half  the weight of the root $r$ then this child  is added to $S$ as $v_{i+1}$. (Note, there can be at most one such child of $v_i$.)
Otherwise, the construction of $S$ stops.
(Note that the spine may consist of only one node, namely, the root of $T$.)
 Let $v_1,v_2,\cdots,v_{s}$ be the nodes of the spine $S$ (Node $v_1$ is the root $r$, and $v_s$ is the last node added to the spine).
 It follows from the definition that if $1\leq i<j\leq s$ then $v_i$ is a strict ancestor of $v_j$. The {\em size} of the spine $S$ is~$s$. 
 We split the nodes  of the spine $S$ to two types. Specifically, the root of $T$, namely $v_1$, is called the {\em apex} node, while all other spine nodes, namely, $v_2,v_3,\cdots,v_s$, are called {\em heavy} nodes. (Recall that the weight of each heavy node is larger than half the weight of the apex node).
 
 By removing the nodes in the spine $S$ (and the edges connected to them),
 the tree~$T$ breaks into $s$ forests $F_1, F_2,\cdots, F_{s}$, such that the following properties holds for each $1\leq i\leq s$:

 \begin{itemize}
 \item
{\bf P1.} In $T$, the roots of the trees in $F_i$ are connected to $v_i$;
 \item
{\bf P2.} Each tree in $F_i$ contains at most $|T|/2$ nodes;
 \item
{\bf P3.} The forests $F_i$ are unrelated in terms of the ancestry relation in $T$.
 \end{itemize}

\begin{figure}
\begin{center}
\includegraphics[width=.5\linewidth]{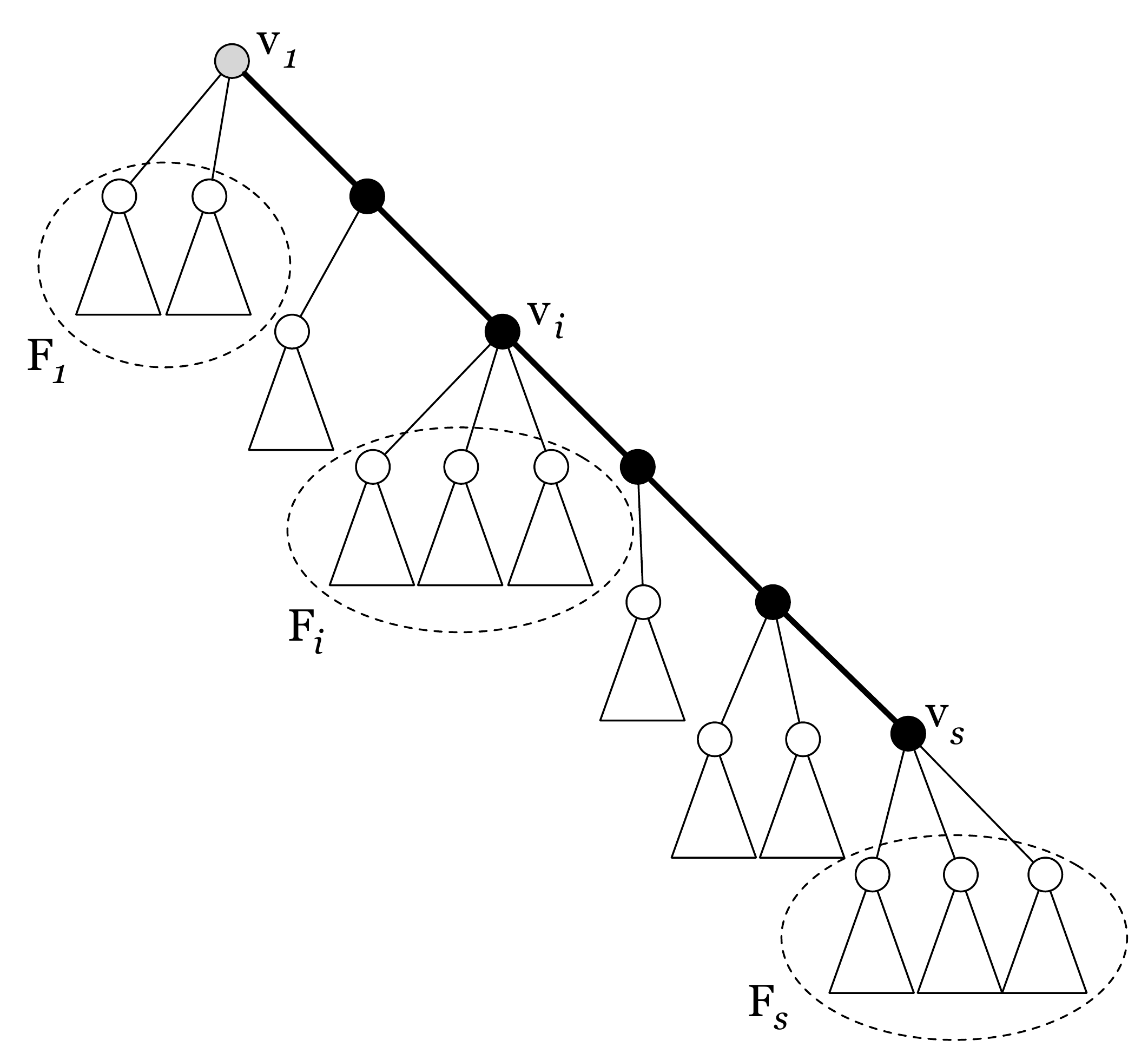}
\caption{Spine decomposition}
\label{fig:spine}
\end{center}
\end{figure}

The spine decomposition is constructed iteratively, where each level of the process follows the aforementioned description. That is, given a forest $F$, after specifying the spine $S$ of each tree $T$ in $F$, we continue to the next level of the process, operating in parallel on the forests  $F_1, F_2,\cdots, F_{s}$.  The recursion  implies that each node is eventually classified as either apex or heavy. 
The {\em depth} of the spine decomposition of a forest $F$, denoted $\SD(F)$ is the maximal size of a spine, taken over all spines obtained in the spine decomposition of $F$. Note that  $\SD(F)$ is bounded from above by the depth of $F$.

For any two integers $n$ and $d$, let $\cF(n,d)$ denote the set of all rooted forests with at most $n$ nodes, whose spine decomposition depth is at most~$d$.

\subsection{Ancestry labeling schemes}\label{subsec:consistent}

An {\em ancestry-labeling
scheme} $( \cM,\cD )$  for a family $\cF$  of forests of rooted trees is
composed of the following components:
\begin{enumerate}
\item A {\em marker} algorithm $\cM$ that
assigns labels (i.e., bit strings) to the nodes of all forests in $\cF$.
\item A  {\em
decoder} algorithm $\cD$ that given any two labels $\ell_1$ and $\ell_2$ in the output domain of $\cM$, returns a boolean value $\cD(\ell_1,\ell_2)$.
\end{enumerate}

These components must satisfy that if $L(u)$ and $L(v)$ denote the labels assigned by the marker algorithm  to
two nodes $u$ and $v$ in some rooted forest  $F\in\cF$, then
\[
\cD(L(u),L(v))=1 \iff \mbox{$u$ is an ancestor of $v$ in $F$.}
\]
It is important to note that the decoder algorithm $\cD$ is independent of the forest $F$. That is, given the labels of two nodes,
the decoder algorithm decides the ancestry relationship between the corresponding nodes without knowing to which forest in $\cF$ they belong.

The most common complexity measure used for evaluating an ancestry-labeling scheme is the {\em label size}, that is,
the maximum number of bits in a label assigned by
$\cM$, taken over all nodes in all  forests in $\cF$.
When considering the {\em query time} of the decoder algorithm, we use the RAM model of computation, and assume that the length of a computer
word is $\Theta(\log n)$ bits.
Similarly to previous works on ancestry-labeling schemes, our decoder algorithm uses only basic RAM operations (which are assumed to take constant time).
Specifically, the basic operations used by our decoder algorithm are the following:
addition, subtraction, multiplication, division,  left/right shifts, less-than comparisons, and extraction of the index of the least significant 1-bit.

Let $\cF$  be a family of forests of rooted trees.
We say that an ancestry-labeling scheme $(\cM,\cD)$ for  $\cF$  is {\em consistent}
if  the decoder algorithm $\cD$ satisfies the following conditions, for any three pairwise different labels $\ell_1,\ell_2$ and $\ell_3$ in the output domain of $\cM$:
\begin{itemize}
\item
{\em Anti-symmetry}: if $\cD(\ell_1,\ell_2)=1$ then $\cD(\ell_2,\ell_1)=0$, and
\item
{\em Transitivity}: if $\cD(\ell_1,\ell_2)=1$ and $\cD(\ell_2,\ell_3)=1$ then $\cD(\ell_1,\ell_3)=1$.
\end{itemize}

Note that by the definition  of an ancestry-labeling scheme $(\cM,\cD)$, the decoder algorithm  $\cD$ trivially satisfies the two conditions above if
 $\ell_i=L(u_i)$ for $i=1,2,3$, and $u_1$, $u_2$ and $u_3$ are different nodes belonging to the same forest in $\cF$.

\subsection{Small universal posets}

The {\em size} of a partially ordered set (poset) is the number of elements in it.
A poset $(X,\leq_X)$ contains a poset $(Y,\leq_Y)$ as an {\em induced suborder}  if there exists an injective mapping
$\phi:Y\to X$ such that for any two elements $a,b\in Y$: we have $$a\leq_Y b \iff \phi(a)\leq_X \phi(b).$$
A poset $(X,\leq)$ is called {\em universal} for a family of posets $\cP$ if $(X,\leq)$  contains every  poset in $\cP$ as an induced suborder. 
If $(X,\leq)$ and $(X,\leq')$  are orders on the set $X$, we say that $(X,\leq')$
is an {\em extension} of $(X,\leq)$ if, for any two elements $x,y\in X$,  $$x\leq y \;  \Longrightarrow  \; x\leq' y.$$
A common way to characterize a poset $(X,\leq)$ is by its {\em  dimension}, that is,
 the smallest number of linear (i.e., total order) extensions
 of $(X,\leq)$  the intersection of which gives rise to $(X,\leq)$ \cite{T92}.  The following fact is folklore, and its proof straightforward (this proof is however stated for the sake of completeness). 
 
  \begin{fact}\label{fact:poset}
 The smallest size of a universal poset for the family of $n$-element posets with dimension at most $k$
 is at most $n^k$.
 \end{fact}
 
 \begin{proof}
 Let $\leq$ be the natural total order defined on the set of integers. We present a universal poset $(U,\preceq)$ for the family of $n$-element posets with dimension at most $k$. The set of elements $U$ is 
 $$U=[1,n]^k=\{u=(u_1,u_2,\cdots,u_k)\mid u_i \in [1,n] \;\mbox{for all}\;  i\in [1,k]\},$$
 and the relation $\preceq$ is defined for two elements $u,v \in U$ by:
 $$
 u \preceq v \iff  u_i\leq v_i, \; \forall i\in[1,k].
 $$
 Clearly $U$ has~$n^k$ elements.
 Now consider any $n$-element  poset $(X,\unlhd)$ with dimension at most $k$. 
 For $i\in [1,k]$, let $(L_i,\leq_i)$ be the  total orders the intersection of which gives rise to $(X,\unlhd)$.
 By the definition of intersection, there exists a collection of injective mappings $\psi_i: X\rightarrow L_i$ such that for any two elements $x,y \in X$, we have $$x \unlhd y \iff \psi_i(x) \leq_i \psi_i(y), \; \forall  i\in [1,k].$$
 For every $i\in[1,k]$, since $(L_i,\leq_i)$ is a  total order, it is isomorphic to $([1,n],\leq)$, that is, there exists an injective and onto mapping $\phi_i:L_i\rightarrow [1,n]$ such that for $a,b\in L_i$, we have $$a\leq_i b \iff \phi_i(a)\leq \phi_i(b).$$ We define the mapping $f:X\rightarrow U$ so that for any $x \in X$, we have the $i$th coordinate $f(x)_i\in [1,n]$ of $f(x)$ be defined as $f(x)_i=\phi_i\circ\psi_i (x)$. The fact that~$f$ preserves the order $\unlhd$, i.e., the fact that, for every $x,y\in X$, $$x\unlhd y \iff  f(x) \preceq f(y).$$ is now immediate.
  \end{proof}

  Another way of characterizing a poset $(X,\leq)$ is by its \emph{tree}-dimension.
A poset $(X,\leq)$ is a {\em tree\footnote{Note that the term ``tree" for ordered sets is used in various senses in the literature, see e.g., \cite{TM77}.}} \cite{B93,W62}
if, for every pair $x$ and $y$ of incomparable elements in $X$, there does not exist an element $z\in X$ such that
$x\leq z$ and $y\leq z$. (Observe that the Hasse diagram  \cite{T92} of a tree poset is a forest of rooted trees).
 The \emph{tree-dimension} \cite{B93} of a poset $(X,\leq)$ is
the smallest number of tree extensions of $(X,\leq)$ the intersection of which gives rise to  $(X,\leq)$.

For any two positive integers $n$ and $k$, let $\cP(n,k)$ denote the family of all $n$-element (non-isomorphic) posets with tree-dimension at most $k$.
The following fact follows rather directly from previous work.

\begin{fact}
Fix an integer $k$ and let
 $M(n)$ denote the smallest size of a universal poset for $\cP(n,k)$.
We have $n^{k-o(1)} \leq M(n)\leq n^{2k}$. 
\end{fact}

\begin{proof}
The fact that the smallest size of a universal poset  for $\cP(n,k)$  is at most $n^{2k}$ follows from Fact \ref{fact:poset}, and from the
well known fact that the dimension of a poset is  at most twice its tree-dimension\footnote{This follows from the fact that a tree-poset $T=(X,\preceq)$ has dimension at most 2. Indeed, consider the  two linear orders for $T$ obtained as follows. We perform two DFS traversals over the Hasse  diagram of $T$, which is a directed forest~$F$, starting from the root  in each tree in $F$, so that to provide every element $x$ with two DFS numbers, $\dfs_1(x)$ and $\dfs_2(x)$. DFS$_1$ is arbitrary,  and DFS$_2$ reverses the order in which the trees are considered  in DFS$_1$, and in which the children are visited in DFS$_1$, so that $x\preceq y$ if and only if $\dfs_1(x)\leq \dfs_1(y)$ and $\dfs_2(x)\leq \dfs_2(y)$.}.
For the other direction, Alon and Scheinerman showed
 that
 the number of non-isomorphic $n$-element posets with dimension at most $k$ is at least $n^{n(k-o(1))}/n!$~(this result is explicit in the proof of Theorem 1 in \cite{AS88}).
Since the dimension of a poset is at least its tree-dimension, this result of
\cite{AS88} yields also a lower bound  on 
 the number of non-isomorphic $n$-element posets with tree-dimension at most $k$, specifically, we have $$n^{n(k-o(1))}/n!\leq |\cP(n,k)|~.$$ On the other hand,   $$|\cP(n,k)|\leq {M(n) \choose n}$$ by definition of $M(n)$. 
Therefore, by combining the above two inequalities, it directly follows that $M(n)\geq n^{k-o(1)}$.
\end{proof}

\section{Labeling schemes for forests with bounded spine decomposition depth}\label{sec:bounded} 

In this section we construct an efficient ancestry-labeling scheme for forests with bounded  spine decomposition depth.
Specifically, for forests with spine decomposition depth at most $d$, our scheme enjoys label size of $\log n +2\log d +O(1)$ bits. (Note that the same bound holds also for forests with depth at most $d$.) Moreover, our scheme has $O(1)$ query time and $O(n)$ construction time. 

\subsection{Informal description}

Let us first explain the intuition behind our construction. Similarly to the simple ancestry scheme in~\cite{KNR92}, we 
map the nodes of forests  to a  set of intervals $\cI$, in a way that relates the ancestry relation in each forest with the partial order defined on intervals through containment. I.e., a label of a node is simply an interval, and the decoder decides the ancestry relation between two given nodes using the interval containment test on the corresponding intervals.  
While the number of intervals used for the scheme in \cite{KNR92} is $O(n^2)$, we managed to show that, if we restrict our attention to forests with spine decomposition depth bounded by $d$, then one can map the set of such forests to a set of intervals $\cI$, whose size is only $|\cI|=O(nd^2)$. Since a label is a pointer to an interval in~$\cI$, the bound of $\log n +2\log d +O(1)$ bits for the label size follows. In fact, we actually manage to  provide an explicit description of each interval, still using $\log n +2\log d +O(1)$ bits, so that to achieve constant query time. 

\subsubsection{Intuition}

Let $\cF(n,d)$ be the family  of all forests  with at most $n$ nodes and spine decomposition depth at most $d$. The challenge of mapping the nodes of forests in $\cF(n,d)$ to a small set of intervals~$\cI$ is tackled recursively, where the recursion is performed over the number of nodes. That is, for $k=1,2,\cdots,\log n$, 
level $k$ of the recursion deals with forests of size at most $2^k$. When handling the next level of the recursion, namely level $k+1$, the difficult case is when we are given a forest $F$ containing a tree $T$ of size larger than $2^k$, i.e., $2^k<|T|\leq 2^{k+1}$.  Indeed,  trees in $F$ of size at most $2^k$ are essentially handled at level $k$ of the recursion. To map the nodes of Tree $T$, we use the spine decomposition (see Subsection~\ref{sub-spine}).

Recall the spine $S=(v_1,\dots,v_s)$ of $T$ and the forests $F_1,F_2,\cdots, F_s$, obtained by removing $S$ from $T$.
Broadly speaking, Properties P2 and P3  of the spine decomposition give hope that the forests $F_i$, $i=1,2,\cdots,s$, could be mapped relying on the previous level $k$ of the recursion. 
Once we guarantee this, we map the $s$ nodes of the spine $S$ in a manner that respects the ancestry relations. 
That is, the interval associated with a spine node $v_i$ must contain all intervals associated with descendants of $v_i$ in $T$, which are, specifically, all the spine nodes $v_j$, for $j> i$, as well as all nodes in $F_j$, for~$j\geq i$.
Fortunately, the number of nodes on the spine is $s\leq d$, hence we need to deal with only few such nodes. 

The intervals in $\cI$ are classified into $\log n$ {\em levels}. These interval levels correspond to the levels of the recursion in a manner to be described. 
Level $k$ of the recursion  maps forests (of size at most~$2^k$) into $\cI_k$, the set of intervals of level at most $k$. In fact, even in  levels of recursion higher than $k$, the nodes in forests containing only trees of size at most $2^k$ are mapped into $\cI_k$. (In particular, a forest consisting of $n$ singleton nodes is  mapped into $\cI_1$.) 
Regarding level $k+1$, a forest of size at most $2^{k+1}$ contains at most one  tree~$T$, where $2^k<|T|\leq 2^{k+1}$. In such a case, the nodes on the spine $S$ of $T$ are mapped to level-$(k+1)$ intervals, and the forests $F_1,F_2,\cdots, F_s$ are mapped to $\cI_k$. 

As mentioned before, to have the ancestry relations in $T$ correspond to the inclusion relations in $\cI$, the level-$(k+1)$ interval $I(v_i)$ to which 
some spine node $v_i$ is mapped, must contain the intervals associated with nodes which are descendants of $v_i$ in $T$. In particular, $I(v_i)$ must contain all the intervals associated with the nodes in the forests $F_i,F_{i+1},\cdots, F_s$. 
Since the number of such intervals is at least $\sum_{j=1}^i |F_i|$ (note, this value can be close to $2^{k+1}$), the length of $I(v_i)$ must be relatively large. Moreover, since level-1 intervals are many  (at least $n$ because they need to be sufficiently many to handle a forest containing $n$ singleton nodes), and since $\cI$ contains all level-$k$ intervals, for $\log n$ values of $k$, we want the number of level-$k$ intervals to decrease with $k$, so that the total number of intervals in $\cI$ will remain small (recall, we would like to have $|\cI|=O(nd^2)$). 
Summing up the discussion above, in comparison with the set of level-$k$ intervals, we would like the set of level-$(k+1)$ intervals
to contain fewer but wider  intervals. 

\begin{example}\label{example}
Let us consider the example depicted in Figure~\ref{fig:example}. We have a tree $T$ of size roughly $2^{k+1}$, with two spine nodes $v_1$ and $v_2$ and two corresponding forests $F_1$ and $F_2$. We would like to map $v_2$ to some interval $I(v_2)$ and map all nodes in $F_2$ to intervals contained in $I(v_2)$. In addition, we would like to map $v_1$ to some interval $I(v_1)$ containing $I(v_2)$, and map all nodes in $F_1$ to intervals contained in $I(v_1)\setminus I(v_2)$. 
\end{example}

\begin{figure}
\begin{center}
\includegraphics[width=0.7\linewidth]{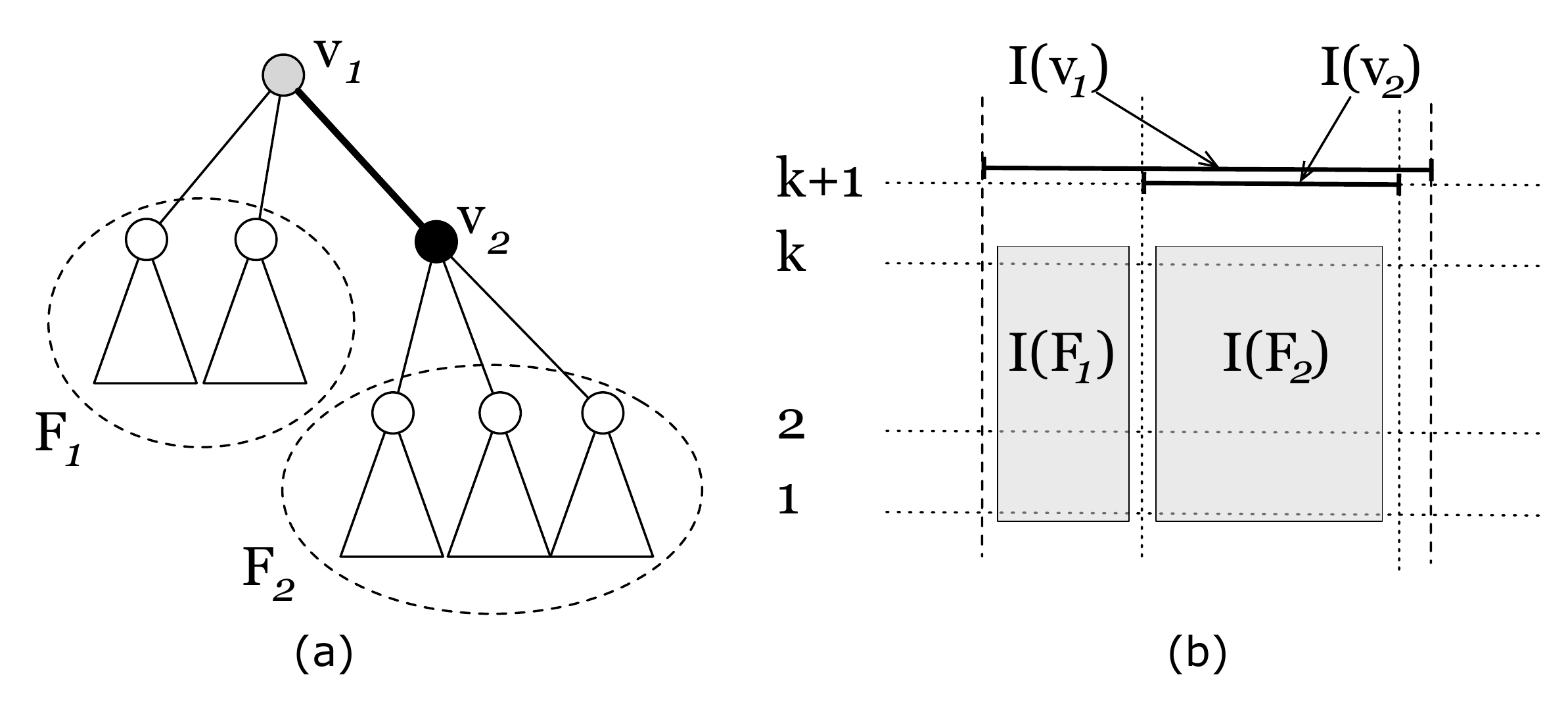}
\caption{Illustration of Example~\ref{example}}
\label{fig:example}
\end{center}
\end{figure}

The  mapping in the above example can be done using the method in \cite{KNR92}. Specifically, $v_1$ is mapped to $I(v_1)=[1,n]$, and $v_2$ is mapped to $I(v_2)=[n-|F_2|+1,n]$. The nodes of $F_2$ are mapped to intervals, all of which are contained in $I(v_2)$, and the nodes of $F_1$ are mapped to intervals which are contained in $I(v_1)\setminus I(v_2)$. Note that this scheme guarantees that all intervals are contained in $[1,n]$. One of the crucial properties making this mapping possible is that fact that the interval $I(v_2)=[n-|F_2|+1,n]$ {\em exists} in the collection of intervals used in \cite{KNR92}, for all possible sizes of $F_2$. Unfortunately,  this property  requires  many intervals of level $k+1$, which is undesirable (the scheme in \cite{KNR92} uses $n^2$ intervals in total). 
In a sense, restricting the number of level-$(k+1)$ intervals  costs us, for example, the inability to use an interval $I(v_2)$ that precisely covers  the set of  intervals associated with $F_2$. In other words, in some cases, $I(v_2)$ must  strictly contain 
$I(F_2):=\bigcup_{v\in F_2} I(v)$.  In particular, we cannot avoid having $|I(v_2)|\geq x+|I(F_2)|$, for some (perhaps large) positive $x$. In addition, the nodes in $F_1$ must be mapped to intervals contained in some range that is outside of $I(v_2)$ (say, to the left of Interval $I(v_2)$), and Node~$v_1$ must be mapped to an interval $I(v_1)$ that contains all these intervals, as well as $I(v_2)$. Hence,  we cannot avoid having $|I(v_1)|\geq x+x'+|I(F_2)|+|I(F_1)|$, for positive integers $x$ and $x'$. Therefore, the total slack (in this case, coming from $x$ and $x'$), does not only propagate over the $s$ spine nodes, but also propagates up the levels. One of the artifacts of this propagation is the fact that we can no longer guarantee that all intervals are contained in $[1,n]$ (as guaranteed by the scheme of \cite{KNR92}). Somewhat surprisingly, we still manage to choose the parameters to guarantee that all intervals in $\cI$ are contained in the range $[1,N]$, where $N=O(n)$.

Being slightly more formal, we introduce a hierarchy of intervals called {\em bins}. A bin $J$ of {\em level $k$} is an interval of length $c_k 2^k$, i.e., $|J|= c_k 2^k$, for some value $c_k$ to be described. Intuitively, the length $c_k 2^k$ corresponds to the smallest length of a bin $J$ for which our scheme enables the proper mapping of any forest of size at most $2^k$ to $J$. It is important to note that this property is {\em shift-invariant}, 
that is, no matter where in $[1,N]$ this bin $J$ is, the fact that its length is at least~$c_k 2^k$ should guarantee that it can potentially contain all intervals associated with a forest of size at most~$2^k$. Because of the aforementioned  unavoidable (non-negligible) slack that propagates up the levels,  we must allow  $c_k$ to increase with $k$. 

\begin{figure}
\begin{center}
\includegraphics[width=0.8\linewidth]{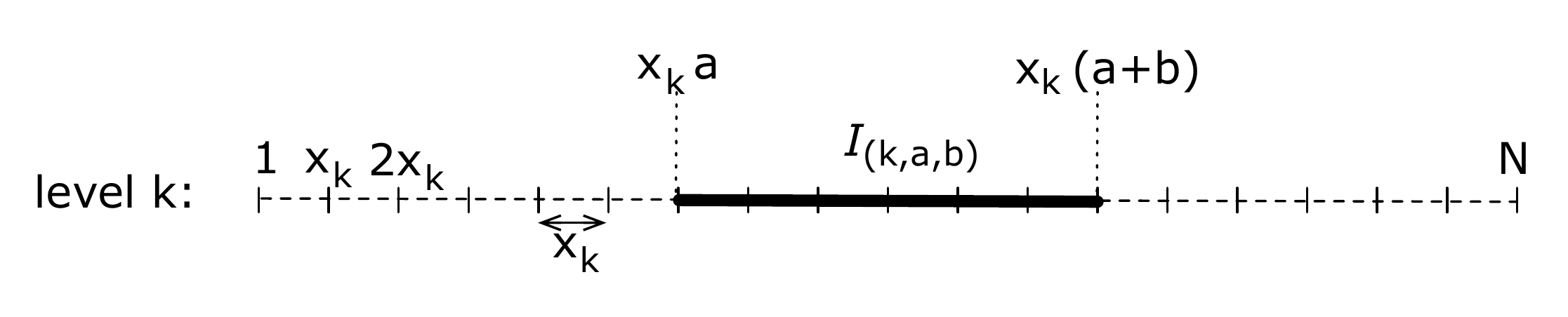}
\caption{A level-$k$ interval $I_{k,a,b}$}
\label{fig:resolution}
\end{center}
\end{figure}

\subsubsection{The intuition behind the tuning of the parameters}

We now describe the set of intervals $\cI$, and explain the intuition behind the  specific choice  of parameters involved. 
Consider a level $k$, and fix a {\em resolution} parameter $x_k$ for interval-level $k$, to be described later. Let $A_k\approx {N}/{x_k}$ and $B_k \approx c_k{2^k}/{x_k}$.
The level-$k$ intervals are essentially all intervals in $[1,N]$ which are of the form:
\begin{equation}\label{eq:I}
 I_{k,a,b}=[a\; x_k,\; (a+b)\; x_k) ~~~\mbox{where}~~~a\in [1,A_k] ~~\mbox{and}~~ b\in [1,B_k].
 \end{equation}
See Figure~\ref{fig:resolution}. The resolution parameter $x_k$ is chosen to be monotonically  increasing with $k$ in a manner that will guarantee fewer intervals of level $k$, as~$k$ is increasing. Moreover, the largest possible length of an interval of level $k$ is $x_k B_k=c_k2^k$, which is the length of a bin sufficient to accommodate the intervals of a tree of size at most $2^k$. This length is monotonically  increasing with the level $k$, as desired.  

\begin{figure}
\begin{center}
\includegraphics[width=0.9\linewidth]{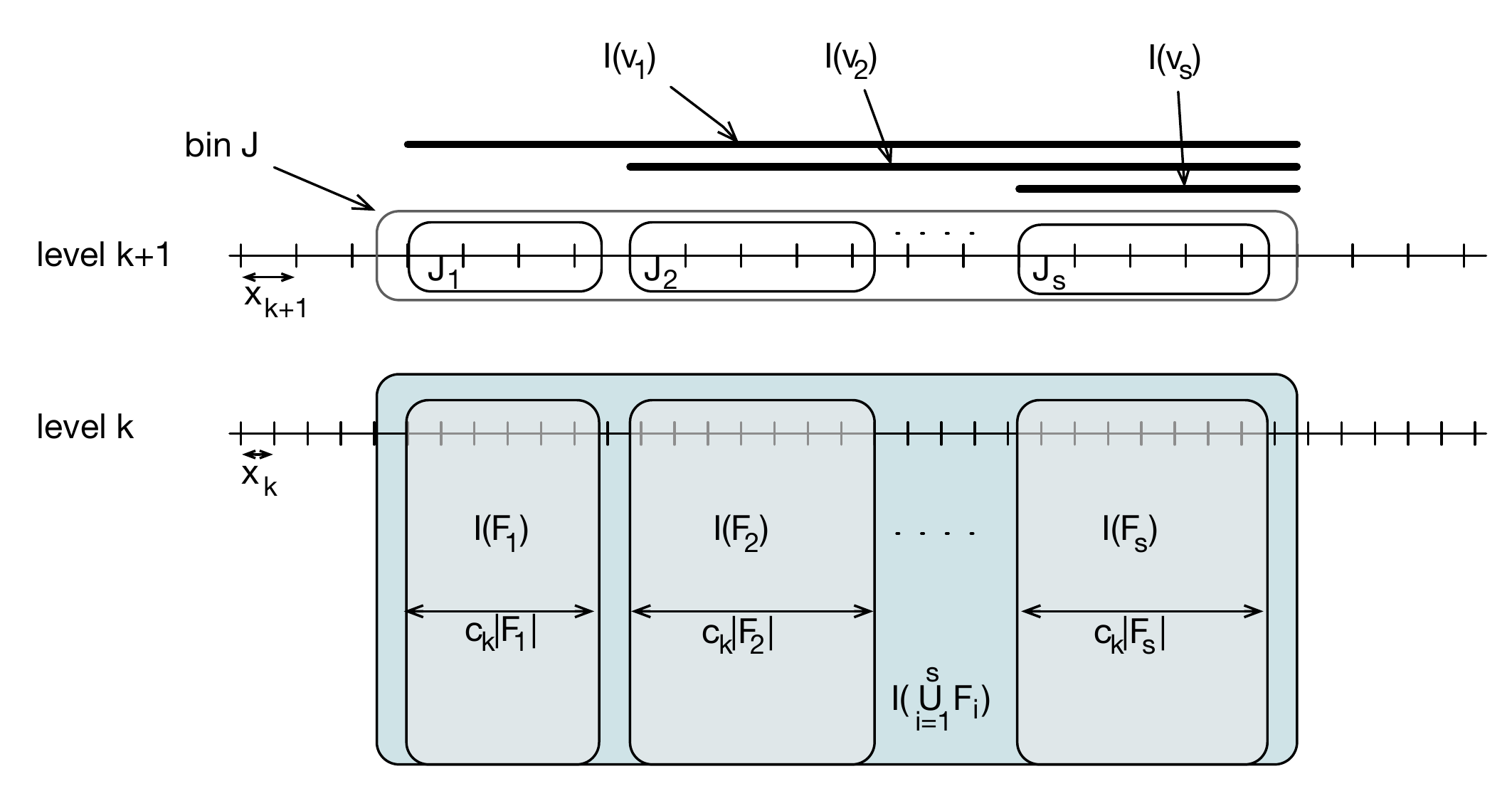}
\caption{Overview of the bins and intervals assignment in level $k+1$}
\label{fig:bins}
\end{center}
\end{figure}

Consider now a bin $J$ of length $c_{k+1}2^{k+1}$ located somewhere in $[1,N]$. This bin $J$ should suffice for the mapping of a tree $T$ of size $2^{k+1}$. By executing the spine decomposition, we obtain the spine nodes $v_1,v_2,\cdots, v_s$ and the forests $F_1,F_2,\cdots F_s$ (see Figure~\ref{fig:spine}). We allocate a level-$(k+1)$ interval $I(v_i)$ to each spine node $v_i$, and a bin $J_i\subseteq J$ to each forest $F_i$, $i=1,\dots,s$, in the same spirit as we did in the specific Example~\ref{example} (see Figure~\ref{fig:example}). 

The general allocation is illustrated in Figure~\ref{fig:bins}. Since $I(v_1)$ is of the form $I_{k+1,a,b}$, and should be included in Bin $J$, and since this interval $I(v_1)$ must contain all intervals assigned to nodes in~$F_1$, Bin $J_1$ is chosen to start at the leftmost multiple of $x_{k+1}$ in Bin $J$. Note that $F_1$ contains trees of size at most $2^k$ each. Hence, by induction on  $k$, each of these trees, $T'$, can be properly mapped to any bin of size $c_k|T'|$. Therefore, setting $J_1$ of size $c_k |F_1|$ suffices to properly map all trees in $F_1$. The bin $J_2$, of size $c_k |F_2|$,  is then chosen to start at the leftmost multiple of $x_{k+1}$ to the right of the end point of $J_1$, in the bin $J$. And so on:  we set $J_i$ of size $c_k |F_i|$, and place it in $J$ so that to start at the leftmost multiple of $x_{k+1}$ to the right of the end point of $J_{i-1}$, $1<i\leq s$. The level-$(k+1)$ intervals associated to the spine nodes are then set as follows. For $i=1,\dots,s$, the interval $I(v_i)$ starts at the left extremity of $J_i$ (which is a multiple of the resolution $x_{k+1}$). All these intervals end at the same point in $[1,N]$, which is chosen as the  leftmost multiple of $x_{k+1}$ to the right of~$J_s$, in~$J$. Putting the right end-point of $J$ at the point where all the intervals of spine nodes end, suffices to guarantee that $J$ includes all the intervals $I(v_i)$, and all the bins $J_i$, for $i=1,\dots,s$. 

Observe that the length of $I(v_s)$ must satisfy $|I(v_s)|\approx c_k|F_s|+x_{k+1}$, where the slack of $x_{k+1}$ comes from the fact that the interval must start at a multiple of the resolution $x_{k+1}$. More generally, for $1\leq i < s$,  the length of $I(v_i)$ must satisfy $$|I(v_i)|\approx c_k |F_i|+x_{k+1}+|I(v_{i+1})|\approx c_k(\sum_{j=i}^s|F_i|)+i \cdot x_{k+1}.$$ Therefore, the length of  $I(v_{1})$ must satisfy $|I(v_1)|\approx c_k(\sum_{i=1}^s|F_i|)+s\cdot x_{k+1}\approx c_k\cdot 2^{k+1}+ s\cdot x_{k+1}$. Now, since $J$ may start at any point between two multiples of the resolution $x_{k+1}$, we eventually get that setting the bin $J$ to be of length $|J| \approx c_k\cdot 2^{k+1}+ (s+1) \cdot x_{k+1}$ suffices. Since $s$ can be at most the spine decomposition depth $d$, we must have $|J|$ be approximately $c_{k+1}2^{k+1} ~\approx~ c_k 2^{k+1}+d\cdot x_{k+1}$. 
 To agree with the latter approximation, we choose the values of $c_k$ so that:
 \begin{equation}\label{eq:ci}
c_{k+1}-c_k ~~\approx~~  \frac{d\cdot x_{k+1}}{2^{k+1}}.
\end{equation}
Ultimately, we would like to map that whole $n$-node forest to a bin of size $c_{\log n} \cdot n$. This bin must fit into $[1,N]$, hence, the smallest value $N$ that we can choose is  $c_{\log n}\cdot n$.
Since we also want the value of $N$ to be  linear in $n$, we choose the $c_k$'s so that $c_{\log n}=O(1)$. Specifically, for $k=1,2,\cdots,\log n$, we set 
$$
c_k \approx \sum_{j=1}^k \frac{1}{j^{1+\epsilon}}
$$
for some small $\epsilon>0$.
Note that $c_k\geq 1$ for each $k$, and the $c_k$'s are increasing with $k$. Moreover, we take $\epsilon$ large enough so that the sum $\sum_{j=1}^\infty 1/j^{1+\epsilon}$ converges. Hence, all the $c_k$'s are bounded from above by some constant~$\gamma$. In particular, $c_{\log n}\leq \gamma$, and thus $N=O(n)$. The fact that all $c_k$'s are bounded, together  with Equation \ref{eq:ci}, explains why we choose $$x_k\approx \frac{2^k}{d\cdot k^{1+\epsilon}}~.$$
This choice for the resolution parameter $x_k$ implies that the number of level-$k$ intervals is $$O(A_k\cdot B_k)=O(nd^2 k^{2(1+\epsilon)}/2^k),$$ yielding a total of $O(nd^2)$ intervals in $\cI$, as desired. In fact, in order to reduce the label size even further, by playing with the constant hidden in the big-$O$ notation, we actually choose $\epsilon$ less than a constant. Indeed, we   will later pick $$c_k\approx \sum_{j=1}^k \frac{1}{j\log^2 j }\;\;\mbox{and}\;\; x_k\approx \frac{2^k}{d\cdot k\log^2 k}~.$$

\subsection{The ancestry scheme}
We now turn  to formally describe the desired ancestry-labeling scheme $(\cM,\cD)$ for $\cF(n,d)$. For simplicity,  assume without loss of generality that  $n$ is a power of 2.

\subsubsection{The marker algorithm $\cM$}\label{section:marker}
We begin by defining the set $\cI=\cI(n,d)$ of intervals.
For integers $a,b$ and $k$, let
$$ I_{k,a,b}=[a\;x_k,\;(a+b) \; x_k)$$
where $$ x_1=1 \;\; \mbox{and} \;\; x_k= \left\lceil\frac{2^{k-1}}{(d+1)k\log^2 k}\right\rceil \;\;\mbox{for $k>1$.}$$
For integer $1\leq k \leq \log n$, let $c_k$ be defined as follows.  Let $c_1=1$, and, for any $k$, $2< k \leq \log n$,
let $$c_k=c_{k-1}+1/k\log^2 k = 1+\sum_{j= 2}^k 1/{j\log^2 j}~.$$
Note that the sum $\sum_{j\geq 2} 1/{j\log^2j}$ converges, and hence all the $c_k$'s are bounded from above by some constant 
$$\gamma= 1+\sum_{j= 2}^\infty 1/{j\log^2 j}~.$$ Let us set:
$$
N=\gamma n.
$$
Then let $A_1=N$, $B_1=0$, and, for $k=2,\dots,\log n$, let us set
\begin{eqnarray*}
A_k =  1+ \left\lceil  \frac{N  (d+1)  k \log^2 k}{2^{k-1}} \right\rceil
\;\;\;\mbox{and} \;\;\;
B_k  = \lceil 2 c_k (d+1)    k \log^2 k\rceil
 \end{eqnarray*}
Next, define the set of level-$k$ intervals:
\[
\cS_k=\{I_{k,a,b}\mid  a\in [1,A_k], \mbox{~and~}~ b \in [1,B_k]\}.
\]
Finally, define the set of intervals of level at most $k$ as
\[
\cI_k=\bigcup_{i=1}^k \cS_i~,
\]
and let
\[
\cI=\cI_{\log n}~.
\]





\begin{definition}\label{def:legal}
Let $F\in\cF(n,d)$. We say that a one-to-one mapping $I:F\rightarrow \cI$ is a {\em legal-containment} mapping if,
for every two nodes $u,v\in F$, we have
\[\mbox{$u$ is an ancestor of $v$ in $F$} \iff  I(v)\subseteq I(u).\]
\end{definition}

Note that since a legal-containment mapping is one-to-one, we get that if $u$ is a strict ancestor of $v$ in $F$, then $I(v)\subset I(u)$, and vice-versa.



We first wish to show that there exists a legal-containment mapping from every forest in $\cF(n,d)$
into $\cI$. For this purpose, we introduce the concept of a {\em bin}, which is simply an interval of integers. For a bin $J$, and  for any integer $k$, $1\leq k \leq \log n$, we use the following notation: 
$$\cI_k(J)= \left\{I_{i,a,b}\in \cI_k\mid  I_{i,a,b} \subseteq J\right\}.$$
I.e., $\cI_k(J)$ is the set of all intervals of level at most $k$ which are contained in the bin $J$.

\begin{claim}\label{claim:obs2}
Let $F$ be a forest, and let
$F_1,F_2, \cdots, F_t$ be  pairwise-disjoint forests such that $\cup_{i=1}^t F_i=F$.  Let $J$ be a bin and let
$J_1,J_2, \cdots, J_t$  be a partition of $J$ into~$t$ pairwise-disjoint bins, i.e.,
 $J=\cup_{i=1}^t J_i$ with $J_i\cap J_j=\emptyset$ for any $1\leq i < j \leq t$. For any level $k$,  $1\leq k\leq\log n$,
if there exists a legal-containment mapping from $F_i$ to $\cI_k(J_i)$ for every $i$, $1\leq i\leq t$, then there
exists a legal-containment mapping from $F$ to $\cI_k(J)$.
\end{claim}

\begin{proof}
The proof follows directly from the definitions above. More specifically, for every integer~$i$, $1\leq i\leq t$, we embed the forest $F_i$ into $\cI_k(J_i)$
using a legal-containment mapping. For two nodes $v$ and $u$ in the same forest $F_i$, the condition specified in Definition \ref{def:legal}, namely, that $u$ is an ancestor of $v$ in $F$ if and only if $I(v)\subseteq  I(u)$, holds by the fact that each $F_i$ is embedded using a legal-containment mapping. On the other hand, if $v$ and $u$ are in two different forests $F_i$ and $F_j$, then the condition holds simply because $J_i\cap J_j=\emptyset$.
\end{proof}

We are now ready to state the main technical lemma of this section.
\begin{lemma}\label{lem:main-bounded}
For every $k$, $1\leq k\leq \log n$, every forest $F\in\cF(2^{k},d)$, and every bin $J\subseteq[1,N)$, 
such that $|J|=\left\lfloor c_k|F|\right\rfloor$,
there exists a legal-containment mapping from $F$ into $\cI_k(J)$. Moreover this mapping can be computed in $O(|F|)$ time.
\end{lemma}
\begin{proof}
We prove the lemma  by induction on $k$. The case $k=1$ is simple and can be verified easily.
Assume now that the claim holds for $k$ with $1\leq k< \log n$, and let us show that it also
holds for $k+1$.
Let $F$ be a forest  of size $|F|\leq 2^{k+1}$, and let $J\subseteq[1,N)$ be  a bin, such that $|J|=\left\lfloor c_{k+1}|F|\right\rfloor$. Our goal is to show that
there exists a legal-containment mapping of $F$ into $\cI_{k+1}(J)$. We consider two cases.

\noindent $\bullet$ {\bf The simpler case:} when all the trees in $F$ are of size at most
$2^{k}$. For this case, we show that  there exists a legal-containment mapping of $F$ into $\cI_k(J)$ for every bin $J\subseteq[1, N)$
such that $|J|=\left\lfloor c_k|F|\right\rfloor$. (Note that this claim is slightly stronger than what is stated in Lemma~\ref{lem:main-bounded})\footnote{Indeed, we show that the size of Bin $J$ can be only 
$\left\lfloor c_k|F|\right\rfloor$, which is  smaller than $\left\lfloor c_{k+1}|F|\right\rfloor$ (that is, the size required to prove the lemma) by an additive term of $-|F|/(k+1)\log^2(k+1)$ .}.

Let $T_1,T_2,\cdots T_t$ be an arbitrary enumeration of the trees in $F$.
We  divide the given bin  $J$ of size $\left\lfloor c_k|F|\right\rfloor$ into $t+1$ disjoint sub-bins $J=J_1\cup J_2\cdots \cup J_t \cup J'$,
where $|J_i|=\left\lfloor c_k|T_i|\right\rfloor $ for every $i$,  $1\leq i\leq t$.
This can be done because $\sum_{i=1}^{t} \left\lfloor c_k|T_i|\right\rfloor \leq \left\lfloor c_k|F|\right\rfloor=|J|$.
By the induction hypothesis, we have a legal-containment mapping of  $T_i$ into $\cI_k(J_i)$ for every $i$, $1\leq i\leq t$.
The stronger claim thus follows by Claim~\ref{claim:obs2}.

Observe that, in the above, the enumeration of the trees  $T_1,T_2,\cdots T_t$ in $F$ was  \emph{arbitrary}. In the context of our general scheme described in the next section, it is important to enumerate these trees in a specific order. Once this order is fixed, we can implement the mapping of $F$ by choosing  the disjoint sub-bins $J_1,\dots,J_t$ of $J$,  so that  $J_i$ is ``to the left'' of $J_{i+1}$, i.e., $J_i\prec J_{i+1}$, for $i=1,\dots,t-1$. This will guarantee that all the intervals associated with the nodes in $T_i$ are ``to the left'' of all the intervals associated with a nodes of $T_j$, for every $1\leq i<j\leq t$. We state this observation as a fact, for further reference in the next section. 

\begin{fact}\label{fact1}
Let $\ell$ be a positive integer. 
Let $T_1,T_2,\cdots T_t$ be an arbitrary enumeration of the trees in a forest $F$, all of size at most~$2^{\ell}$, and let $J\subseteq[1, N)$ be a bin with $|J|=\left\lfloor c_\ell |F|\right\rfloor$. Then, our legal-containment mapping from $F$ into $\cI_\ell (J)$ guarantees  that for every $u\in T_i$ and $v\in T_j$ where $j>i$, we have $I(u) \prec I(v)$.
\end{fact}

\noindent $\bullet$ {\bf The more involved case:} when one of the subtrees in $F$, denoted by $\widehat{T}$, contains more than $2^{k}$ nodes. 
Our goal now is to show that for every bin $\widehat{J}\subseteq[1,N)$, where $|\widehat{J}|=\lfloor c_{k+1}|\widehat{T}|\rfloor$,
there exists a legal-containment mapping of $\widehat{T}$ into $\cI_{k+1}(\widehat{J})$.
Indeed, once this is achieved we can complete the proof as follows. Let $F_1=F\setminus \widehat{T}$, and $F_2=\widehat{T}$. 
Similarly to the simple case above,  let $J_1$ and $J_2$ be two consecutive intervals in $J$ (starting at the leftmost point in $J$) such that  $|J_1|=\lfloor c_k|F_1|\rfloor$ and  $|J_2|=|\widehat{J}|$. 
Since we have a legal-containment mapping that maps $F_1$ into $\cI_{k}(J_1)$, and one that maps
$F_2$ into $\cI_{k+1}(J_2)$,
we get the desired legal-containment mapping of $F$ into $\cI_{k+1}(J)$ by Claim~\ref{claim:obs2}. (The legal-containment mapping of $F_1$ into $\cI_{k}(J_1)$ can be done by the induction hypothesis, because $|F_1|\leq 2^k$.)

For the rest of the proof, our goal is thus to prove the following claim:

\begin{claim}
For every tree $T$ of size $2^k<|T|\leq 2^{k+1}$, and every bin $J\subseteq[1,N)$, where $|J|=\left\lfloor c_{k+1}|T|\right\rfloor$,
there exists a legal-containment mapping of $T$ into $\cI_{k+1}(J)$. 
\end{claim}

In order to prove the claim, we use the spine decomposition described in Subsection~\ref{sub-spine}.
Recall the spine  $S=(v_1,v_2,\cdots,v_s)$, and the corresponding forests $F_1,F_2,\cdots, F_s$. 
The given bin $J$ can be expressed as $J=[\alpha,\alpha+\left\lfloor c_{k+1}|T|\right\rfloor)$ for some integer $\alpha < N - \left\lfloor c_{k+1}|T|\right\rfloor$.
We now describe how we allocate the sub-bins $J_1,J_2\dots,J_s$ of $J$ so that, later, we will map each $F_i$ to $\cI_{k}(J_i)$.

\paragraph{The sub-bins $J_1,J_2\dots,J_s$ of $J$:}
For every  $i=1,\dots,s$, we now define a bin $J_i$ associated with~$F_i$. 
Let us first define $J_1$. Let $a_1$ be the smallest integer such that $\alpha\leq a_1x_{k+1}$.
We let
 $$J_1=[a_1x_{k+1},a_1x_{k+1} + \left\lfloor c_{k}|F_1|\right\rfloor)~.$$
  Assume now that we have defined
the interval $J_i=[a_i x_{k+1},a_i x_{k+1}+ \left\lfloor c_{k}|F_i|\right\rfloor)$ for $1\leq i< s$. We define
the interval $J_{i+1}$ as follows.
Let $b_{i}$ be the smallest integer such that $\left\lfloor c_{k}|F_{i}|\right\rfloor\leq b_{i}x_{k+1}$, that is
\begin{equation}\label{eq:bi}
(b_i-1)x_{k+1}< \left\lfloor c_{k}|F_{i}|\right\rfloor\leq b_{i}x_{k+1}~.
\end{equation}
Then, let $a_{i+1}=a_i+b_i$, and define $$J_{i+1}=[a_{i+1} x_{k+1},a_{i+1} x_{k+1}+ \left\lfloor c_{k}|F_{i+1}|\right\rfloor).$$ 
Hence, for $i=1,2,\cdots, s$, we have $|J_i|=\left\lfloor c_{k}|F_{i}|\right\rfloor$. Moreover, 
\begin{equation}\label{eq:Jiincluded}
J_{i}\subseteq [a_{i} x_{k+1},(a_{i} + b_i)x_{k+1}) = I_{k+1,a_i,b_i}.
\end{equation}
Also observe that, for every $i$, $1\leq i\leq s-1$, we have:
\begin{equation}\label{eq:prec}
J_i\prec J_{i+1}.
\end{equation}

%

 Since $a_1x_{k+1}<\alpha + x_{k+1}$, and since we have a ``gap'' of at most $x_{k+1}-1$ between any consecutive sub-bins $J_i$ and $J_{i+1}$, we get that
\[
 \bigcup_{i=1}^{s} I_{k+1,a_i,b_i} \subseteq \Big [\alpha, \; \alpha+(s+1)(x_{k+1}-1) + \left\lfloor c_{k}|T|\right\rfloor \Big).
\]
Now, since $s\leq d$ and $2^k<|T|$, and since $(d+1) (x_{k+1}-1)\leq \left\lfloor\frac{2^k}{(k+1)\log^2(k+1)}\right\rfloor$, it follows that,
\begin{equation}\label{eq:1}
\bigcup_{i=1}^{s} I_{k+1,a_i,b_i}  \subseteq  \left[\alpha,\alpha + \left\lfloor\frac{|T|}{(k+1)\log^2(k+1)}+c_{k}|T|\right\rfloor\right)
  =  [\alpha,\alpha + \left\lfloor c_{k+1}|T|\right\rfloor)
   =J~.
\end{equation}
Since $\bigcup_{i=1}^{s} J_i \subseteq \bigcup_{i=1}^{s} I_{k+1,a_i,b_i}$, we finally get that 
\[
\bigcup_{i=1}^{s} J_i  \subseteq J.
\]
On the other hand, since, for $1\leq i\leq s$, the length of $J_i$ is  $\left\lfloor c_{k}|F_i|\right\rfloor$, and since
 each tree in $F_i$ contains at most $2^k$ nodes, we get, by the induction hypothesis, that
 there exists a legal-containment mapping of each
$F_i$ into $\cI_k(J_i)$. We therefore  get a legal-containment mapping from  $\bigcup_i^s F_i$ to $\cI_k(J)$, by Claim~\ref{claim:obs2}.

By Equation~\ref{eq:prec}, we have $J_i\prec J_{i+1}$ for every $i$, $i=1,2,\cdots, s-1$. This  guarantees that all the intervals associated with the nodes in $F_i$ are ``to the left'' of all the intervals associated with a nodes of $F_j$, for every $1\leq i<j\leq s$. 
We state this observation as a fact, for further reference in the next section. 

\begin{fact} \label{fact2}
Let $\ell$ be a positive integer. 
Let $F_1,F_2,\cdots F_s$ be the forests of the spine $S=(v_1,v_2,\dots,v_s)$ of the tree $T$ with $2^{\ell-1}<|T|\leq 2^{\ell}$, and let $J\subseteq[1, N)$ be a bin satisfying $|J|=\left\lfloor c_{\ell}|T|\right\rfloor$. Our legal-containment mapping from $T$ into $\cI_{\ell}(J)$ guarantees that, for every $u\in F_i$ and $v\in F_j$ where $j>i$, we have $I(u)\in J_i$ and $I(v)\in J_j$, and hence $I(u) \prec I(v)$.
\end{fact}

It is now left to map the nodes in the spine $S$ into $\cI_{k+1}(J)$, in a way that respects the ancestry relation.

\paragraph{The mapping of spine nodes into level-$(k+1)$ intervals:}

 For every $i$, $1\leq i\leq s$, let $\widehat{b}_i=\sum_{j=i}^s b_j$ where the $b_j$s are defined by Equation~\ref{eq:bi}. 
We map the node $v_i$ of the spine to the interval $$I(v_i)=I_{k+1,a_i,\widehat{b}_i}~.$$
Observe that, by this definition, $$I(v_i)=\bigcup_{j=i}^s I_{k+1,a_i,b_i}~.$$
By Equations~\ref{eq:Jiincluded} and~\ref{eq:1}, we get $$\bigcup_{j=i}^s J_j \subseteq I(v_i)\subseteq J~.$$
To show  that $I(v_i)$ is indeed in $\cI_{k+1}(J)$, we still need to show that $a_i\in [1,A_{k+1}]$ and $\widehat{b}_i\in [1,B_{k+1}]$. Before showing that, let us make the following observation resulting from the combination of Equation~\ref{eq:prec} and Fact~\ref{fact2}, to be used for further reference in the next section. 

\begin{fact} \label{fact3}
Let $\ell$ be a positive integer. 
Let $F_1,F_2,\cdots F_s$ be the forests of the spine $S=(v_1,v_2,\dots,v_s)$ of the tree $T$ with $2^{\ell-1}<|T|\leq 2^{\ell}$, and let $J\subseteq[1, N)$ be a bin satisfying $|J|=\left\lfloor c_{\ell}|T|\right\rfloor$. Our legal-containment mapping from $T$ into $\cI_{\ell}(J)$ guarantees that, for every $u\in F_i$, $i\in\{1,\dots,s-1\}$, and for every $j>i$, we have $I(u) \prec I(v_j)$.
\end{fact}

Let us now show that $I(v_i)$ is
indeed in $\cI_{k+1}(J)$. 
It remains to show that $a_i\in [1,A_{k+1}]$ and that $ \widehat{b}_i \in [1,B_{k+1}]$. 
Recall that $J$ is of the form $J=[\alpha,\alpha+\left\lfloor c_{k+1}|T|\right\rfloor)$ with $1\leq \alpha < N-\left\lfloor c_{k+1}|T|\right\rfloor$. Note that,
$$N~\leq \left\lceil  N~ \frac{(d+1)(k+1)\log^2(k+1)}{2^k}\right\rceil \left\lceil\frac{2^k}{(d+1)(k+1)\log^2(k+1)}\right\rceil~= ~(A_{k+1}-1)x_{k+1}.$$
Therefore, we get that 
\begin{equation}\label{eq:alpha}
\alpha < (A_{k+1}-1)x_{k+1} -\left\lfloor c_{k+1}|T|\right\rfloor.
\end{equation}
On the other hand, by definition of the $a_i$'s, we get that, for every $i$, 
$$a_ix_{k+1} \leq a_1 x_{k+1} + i \cdot x_{k+1} + \sum_{j=1}^i \lfloor c_k |F_i| \rfloor \leq a_1 x_{k+1}+ d \cdot x_{k+1} + \lfloor c_k |T| \rfloor~.$$
Moreover, by the minimality of $a_1$, we have $a_1 x_{k+1} \leq \alpha + x_{k+1}$. Combining the latter two inequalities we get 
$$a_ix_{k+1} \leq \alpha + (d+1)\cdot x_{k+1} + \lfloor c_k |T| \rfloor ~.$$
Combining this with Equation~\ref{eq:alpha}, we get 
$$a_ix_{k+1} \leq A_{k+1} x_{k+1} +
(d \cdot x_{k+1} + \lfloor c_k |T| \rfloor - \left\lfloor c_{k+1}|T|\right\rfloor)~.$$
It follows directly from the definition of $x_{k+1}$ and $c_k$, and from the fact that $|T|>2^k$, that $d \cdot x_{k+1} + \lfloor c_k |T| \rfloor - \left\lfloor c_{k+1}|T|\right\rfloor \leq 0$, implying that $a_i  \leq A_{k+1}$, as desired. 

Let us now show that $ \widehat{b}_i \in [1,B_{k+1}]$. 
Recall that, by definition,   $ \widehat{b}_i \leq \sum_{j=1}^s b_j$ for every $i$, $1\leq i\leq s$. So it is enough to show that $\sum_{i=1}^s b_i\leq B_{k+1}$.  By construction, $$x_{k+1} \sum_{i=1}^s b_i \leq |J| = \lfloor c_{k+1} |T| \rfloor \leq c_{k+1} 2^{k+1}.$$
So, it suffices to show that $c_{k+1}2^{k+1}/x_{k+1}\leq B_{k+1}$. Fortunately, this holds by definition of the three parameters $c_{k+1}$, $x_{k+1}$, and $B_{k+1}$. 

The above discussion shows that for all $i=1,\dots,s$, we have $a_i  \leq A_{k+1}$ and $\widehat{b}_i\leq B_{k+1}$, which implies that $I(v_i)\in \cI_{k+1}(J)$. 

We now show that our mapping is indeed a legal-containment mapping.
Observe first that, for $i$ and $j$ such that $1\leq i<j\leq s$, we have $$I_{k+1,a_i,\widehat{b}_i}\supset I_{k+1,a_j,\widehat{b}_j}.$$
Thus,  $I(v_i)\supset I(v_j)$, as desired.

In addition,
recall  that, for every $j=1,\dots,s$, the forest $F_j$ is mapped into $\cI_k(J_j)$. Therefore, if $I(v)$ is the interval
of some node $v \in F_j$, then  we have $I(v)\subset J_j$. Since
 $J_j\subset I_{k+1,a_i,\widehat{b}_i}$ for every $i$ such that $1\leq i \leq j\leq s$, we obtain
 that $I(v)$ is contained in $I(v_i)$,  the interval associated with $v_i$. This establishes the fact that $I:F\rightarrow \cI_{k+1}(J)$ is a legal-containment mapping. Since the recursive spine decomposition of a forest $F$ takes $O(|F|)$ time, it follows that  this legal-containment mapping  also takes $O(|F|)$ time. This completes the proof of Lemma~\ref{lem:main-bounded}.
\end{proof}

Lemma~\ref{lem:main-bounded} describes the interval assignments to the nodes. We  now  describe the labeling process.

\paragraph{The label assignment:}
By Lemma~\ref{lem:main-bounded},  we get that  there exists a legal-containment mapping $I:F\rightarrow \cI$, for any $F\in\cF(n,d)$.
The marker $\cM$ uses this  mapping to label
the nodes in $F$.
Specifically, for every node $u\in F$, the interval $I(u)=I_{k,a,b}\in \cI$  is encoded in the label $L(u)$ of $u$ as follows.
The first $\lceil\log k \rceil$ bits in $L(u)$  are set to $0$ and the following bit is set to 1. The next
$\lceil\log k \rceil$ bits are used to encode $k$. The  next
$\log d+\log k+2\log\log k+O(1)$ bits are used to encode the value $b\in B_k$, and the next $\log n+\log d+\log k +2 \log\log k -k+O(1)$ bits to encode $a$.
Since $2(\log k +2 \log\log k) -k =O(1)$, we get the following.

\begin{lemma}\label{lem:label-size-bounded}
The marker algorithm $\cM$ assigns the labels of nodes in some forest $F\in\cF(n,d)$ in $O(n)$ time, and each such label is encoded using
 $\log n +2\log d +O(1)$ bits.
\end{lemma}

\subsubsection{The decoder  $\cD$}
Given a label $L(u)$ of some node $u$ in some forest $F\in \cF(n,d)$, the decoder $\cD$  extracts the interval $I(u)=I_{k,a,b}$ as follows. First,
since $F \in \cF(n,d)$, the decoder $\cD$  knows $d$, and thus knows $\log d$. By finding the first bit that equals 1 it can extract the value $\lceil\log k \rceil$. Then by inspecting the next $\lceil\log k \rceil$ bits it extracts the value $k$. 
Subsequently, the decoder  inspects the next $\log d+\log k+2\log\log k  +O(1)$ bits and extracts  $b$. Finally, $\cD$ inspects the remaining $\log n+\log d+\log k+2\log\log k -k+O(1)$ bits in the label to extract $a\in A_k$.
At this point the decoder has $k$, $a$, and $b$ and it can recover $I_{k,a,b}$ using $O(1)$ multiplication 
and division operations. Recall, in the context of this section, we assume that such operations  take constant time. We thus get the following.
\begin{observation}\label{obs:query-bounded}
Given a label $L(u)$ of some node $u$ in some forest $F\in \cF(n,d)$, the decoder $\cD$ can extract $I(u)=I_{k,a,b}$ in constant time.
\end{observation}

Given the labels $L(u)$ and $L(v)$ of two nodes $u$ and $v$ is some rooted forest $F\in\cF(n,d)$, the decoder
finds the ancestry relation between the nodes using a simple  
interval containment test between the corresponding intervals, namely $I(u)$ and $I(v)$.

The fact that the intervals are assigned by a legal-containment mapping ensures the following: 

\begin{lemma}\label{lem:correct}
$(\cM,\cD)$ is a correct ancestry-labeling scheme for $\cF(n,d)$.
\end{lemma}

Lemmas~\ref{lem:label-size-bounded} and~\ref{lem:correct} imply that there exists an ancestry-labeling scheme  for $\cF(n,d)$ with label size at most
$\log n +2\log d +O(1)$ bits. Recall that $\cF(n,d)$ denotes the set of all forests with at most $n$ nodes, and spine-decomposition depth at most $d$, while $\cF(n)$ denotes  the set of all forests with at most $n$ nodes. In general, an ancestry scheme that is designed for a family $\cF$ of forests may rely on a decoder that takes advantage from the fact that the labels are the ones of nodes belonging to a forest $F\in \cF$. For instance, in the case of $\cF=\cF(n,d)$, the decoder can rely on the knowledge of $n$ and $d$.  Although the ancestry-labeling scheme described above was designed for forests in $\cF(n,d)$, we show that a slight modification of it applies to all forests, at a very small cost in term of label size. 
This will establish our main theorem for this section. 

\subsubsection{Main theorem for forests with bounded spine-decomposition depth}

\begin{theorem}\label{thm:anc-bounded}
~
\begin{enumerate}
\item
There exists an ancestry-labeling scheme  for $\cF(n)$
such that any node in a forest of spine decomposition depth at most $d$ is labeled using at most
$\log n +2\log d +O(1)$ bits. 
\item There exists an ancestry-labeling scheme  for the family of all forests
such that any node in a forest of size at most $n$ and spine decomposition depth  at most $d$ is labeled using at most
$\log n +2\log d +2\log\log d+O(1)$ bits.
\end{enumerate}
In both schemes, the query time of the scheme is constant, and the time to construct the scheme for a given forest is linear.
\end{theorem}

\begin{proof}
Lemma~\ref{lem:correct} together with  Lemma~\ref{lem:label-size-bounded} and Observation~\ref{obs:query-bounded}  
establishes the fact that there exists an ancestry-labeling scheme  for $\cF(n,d)$ with label size at most
$\log n +2\log d +O(1)$ bits. Moreover, the query time of the scheme is constant, and the time to construct the scheme for a given forest is linear.

To obtain a scheme for $\cF(n)$, we can no longer assume that the decoder of the scheme knows $d$. However, by applying a simple trick we show that the scheme for $\cF(n,d)$ can be easily transformed to a scheme for $\cF(n)$ with the desired bound.
%
%
We define a labeling scheme $(\cM,\cD)$ for $\cF(n)$.
Given a forest $F\in \cF(n)$ of spine-decomposition depth $d$, let $\widehat{d}=2^{\lceil \log d \rceil}$, i.e., $\widehat{d}$ is the smallest integer larger
than $d$ which is a power of 2. Obviously, $F\in \cF(n,\widehat{d})$.
The first part of the proof tells us that there exists an ancestry-labeling scheme $(\cM_{\widehat{d}},\cD_{\widehat{d}})$  for $\cF(n,\widehat{d})$
which uses labels each composed of precisely $L=\log n +2\log \widehat{d} +O(1)$ bits.
(By padding enough zeros to the left of the label,
we can assume that each label consists of precisely $L$ bits.) Moreover, the labels are assigned in $O(|F|)$ time.
The marker $\cM$ uses this scheme to label the nodes in $F$. 
The decoder $\cD$ operates as follows.
Given the labels of two nodes $u$ and $v$ in $F$, the decoder $\cD$
first finds out what is $\widehat{d}$ (this can be done, since $n$ and $L$ are known to $\cD$, and since $\widehat{d}$  is a power of 2),
and then uses the (constant time) decoder $\cD_{\widehat{d}}$ to interpret the relation between $u$ and $v$. The bound on the size of a label follows as
$L=\log n +2\log d +O(1)$.

Finally, we now show that, with a slight increase on the label size, one can have
an ancestry-labeling scheme for the family  of all forests (i.e., in such a scheme, given the labels of two nodes in a forest $F$,
the decoder doesn't have bounds on neither the size of $F$ nor on its spine-decomposition depth).
Let $F$ be a forest with $n$ nodes and depth $d$.
Let $\widehat{n}=2^{\lceil \log n \rceil}$.
We label the nodes of $F$ using the marker of the scheme $(\cM_{\widehat{n}},\cD_{\widehat{n}})$ for $\cF(\widehat{n})$ mentioned above.
By adding $2\lceil\log\log d\rceil$ bits to the label of each node in $F$,
one can assume that given  a label of a node in $F$, the decoder knows the value of $\log d$.
Now given a label of a node in $F$, the decoder can extract $\log \widehat{n}$ (using the size of the label, in a method similar to
one described above). Since $\widehat{n}=2^{\log \widehat{n}}$, we can assume that
the decoder knows $\widehat{n}$. Thus, to extract the ancestry relation between the two nodes in $F$, the decoder uses $\cD_{\widehat{n}}$.
\end{proof}

Note that the depth of a forest $F$ is bounded from above by the spine decomposition depth of~$F$. Hence, 
all our aforementioned results for forests with bounded  spine decomposition depth hold also restricted to bounded depth forests.
Hence, Theorem~\ref{thm:anc-bounded} translates to the following corollary.

\begin{corollary}\label{cor:ancestry3}~
\begin{enumerate}
\item
There exists an ancestry-labeling scheme  for $\cF(n)$
such that any node in a forest of depth at most  $d$ is labeled using at most
$\log n +2\log d +O(1)$ bits. 
\item There exists an ancestry-labeling scheme  for the family of all forests
such that any node in a forest of size $n$ and depth at most $d$ is labeled using at most
$\log n +2\log d +2\log\log d+O(1)$ bits.
\end{enumerate}
In both schemes, the query time of the scheme is constant, and the time to construct the scheme for a given forest is linear.
\end{corollary}

\subsection{A parenthood-labeling scheme}

The ancestry-labeling scheme described in Corollary \ref{cor:ancestry3} can be advantageously transformed into a parenthood-labeling scheme
which is very efficient for trees of small depth. Recall that an parenthood-labeling
scheme  for the  family of rooted  forests $\cF$ is a pair $(\cM,\cD)$ of marker and decoder, satisfying that if $L(u)$ and $L(v)$ are the labels given
by the marker $\cM$ to
two nodes $u$ and $v$ in some forest $F\in\cF$, then:
$\cD(L(u),L(v))=1 \iff \mbox{$u$ is the parent of $v$  in $F$}$.

Similarly to the ancestry case, we evaluate a parenthood-labeling scheme
$(\cM,\cD)$ by its {\em label size}, namely
the maximum number of bits in a label assigned by the marker
algorithm $\cM$ to any node in any forest in $\cF$.

For  two nodes $u$ and $v$ in a rooted forest $F$, $u$ is a parent of $v$ if and only if  $u$ is an ancestor
of $v$ and $depth(u)=depth(v)-1$. It  follows that one can easily transform any ancestry-labeling scheme for
$\cF(n)$ to a parenthood-labeling scheme for
$\cF(n)$ with an extra additive term of $\lceil\log d\rceil$ bits to the label size (these bits are simply used to encode
the depth of a vertex).
The following theorem follows.
\begin{theorem}\label{thm:adj} ~
\begin{enumerate}
\item
There exists a parenthood-labeling scheme  for  $\cF(n)$
such that any node in a forest of depth at most $d$ is labeled using
$\log n +3\log d +O(1)$ bits. 
\item
There exists a parenthood  scheme  for the family of all rooted forests 
such that any node in an $n$-node forest of depth at most $d$ is labeled using
$\log n +3\log d +2\log\log d+ O(1)$ bits. 
\end{enumerate}
For both schemes, the query time is constant and the time to construct the scheme for a given forest is linear.
\end{theorem}

\section{The general  ancestry-labeling scheme}\label{sec:main}

This section is dedicated to the construction of 
an ancestry-labeling scheme for forests, which has label size $\log n+O(\log\log n)$ bits for $n$-node forests.
Moreover, given an $n$-node forest $F$, the labels can be assigned to the nodes of $F$ in $O(n)$ time,
and any ancestry query can be answered in constant time.

In Section \ref{sec:bounded}, we managed to construct an efficient ancestry labelling scheme for forests with bounded spine decomposition depth. 
 Unfortunately, the situation becomes more difficult if the input forests can have long spines. To handle this case, we introduce a new tree-decomposition, called the {\em folding} decomposition.
Applying this decomposition to a forest $F$ results in a forest~$F^*$, on the same set of vertices, whose spine decomposition is of depth at most 2. Moreover, this transformation partially preserves the ancestry relation in $F$ in the following sense: if $v$ is an ancestor of $u$ in $F^*$ then $v$ is also an ancestor of $u$ in $F$. Next, we apply the ancestry scheme from Section  \ref{sec:bounded} over $F^*$, resulting in labels of size $\log n+O(1)$ for the nodes of $F^*$. Giving the corresponding labels to the nodes of $F$ enables  the decoder to detect all ancestry relations in $F$ that are preserved in $F^*$.
At this point, it remains to deal with ancestry relations that are not preserved by this transformation. For this purpose, we shall provide nodes with additional  information which can be encoded using $O(\log\log n)$ bits per node. To explain how we do so, let us first describe the folding decomposition.

\subsection{The folding decomposition}
\label{sec:folding}

We construct the folding decomposition recursively, according to the recursive construction of the spine decomposition, as described in  Subsection~\ref{sub-spine}. In a given level of the recursion, we are given a tree $T$ with spine $S=(v_1,v_2,\cdots,v_s)$ (see Figure~\ref{fig:folding}).  In the folding decomposition $T^*$ of $T$, the apex node $v_1$ remains the root, and all its children in $T$ remain children of $v_1$ in $T^*$. In addition, all heavy nodes $v_2,v_3,\cdots,v_s$ also become children of $v_1$ in $T^*$. Furthermore, given a heavy node~$v_i$, $2\leq i\leq s$, all of its  non-heavy children  in $T$  remain children of $v_i$ also in $T^*$. The next level of the recursion now applies separately to all trees in all forests $F_1,F_2,\cdots, F_s$ resulting from removing the spine nodes from the tree (see Subsection~\ref{sub-spine}). Note that all these trees have at most half the number of nodes in $T$. Also note that the roots of these trees are going to be apex nodes in the next level of the recursion.

\begin{figure}
\begin{center}
\includegraphics[width=0.9\linewidth]{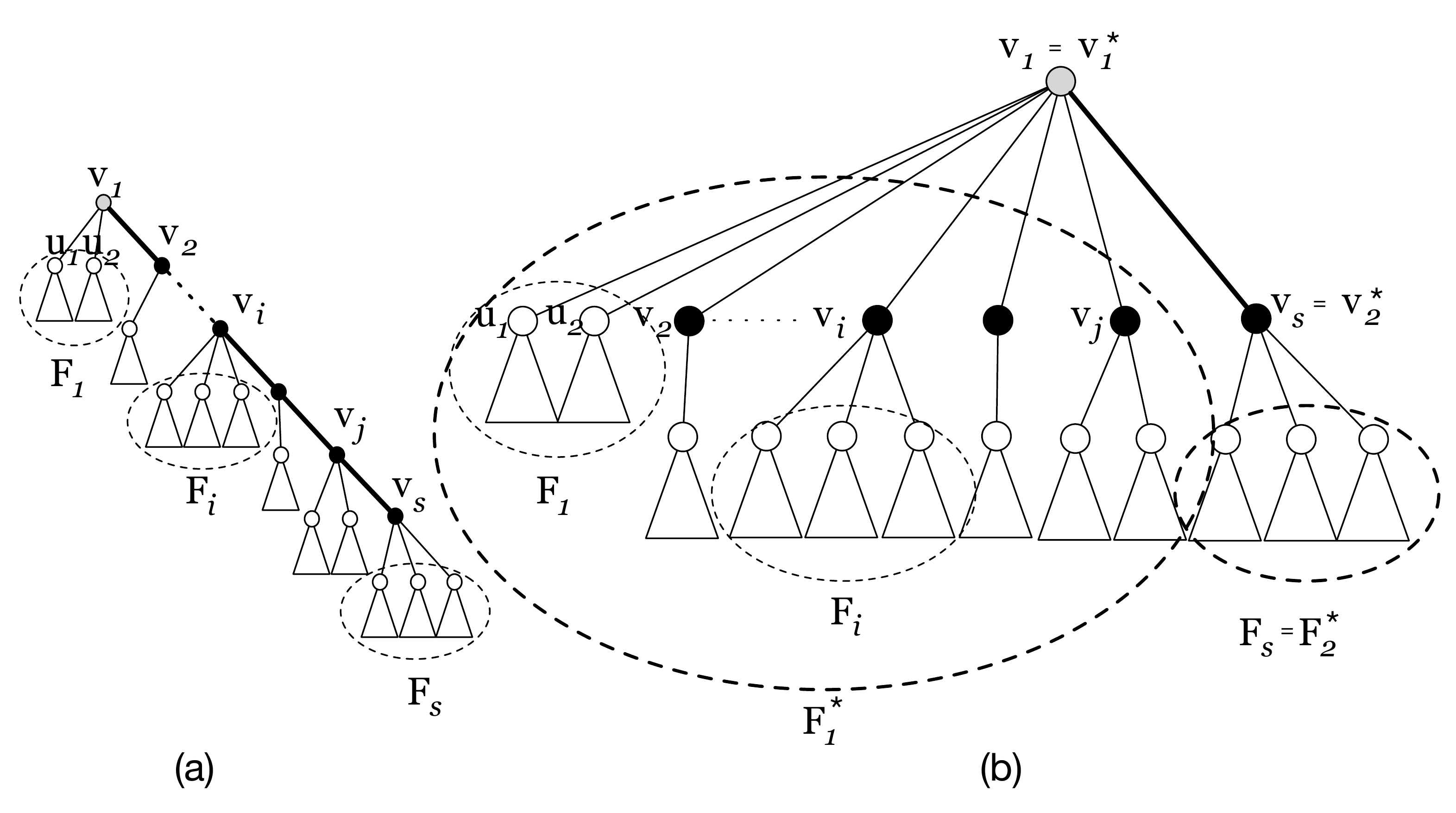}
\caption{Folding decomposition $T^*$ (b) from the spine decomposition of $T$ (a)}
\label{fig:folding}
\end{center}
\end{figure}

The following lemma provides a critical property of the folding decomposition. 

\begin{lemma}\label{lem:depth2}
Let $T$ be a tree, and let $T^*$ be its folding decomposition. 
The spine decomposition depth of $T^*$ is at most~2. Moreover, if $v$ is a strict ancestor of $u$ in $T^*$, then $v$ is an ancestor of $\apex(u)$ in $T^*$. 
\end{lemma}

\begin{proof}
Consider a level $\ell$ of the folding decomposition of $T$, dealing with some subtree $T'$ of $T$ (see Figure~\ref{fig:folding}). 
Let $S=(v_1,\dots,v_s)$ be the spine of $T'$, and recall that $F'_i$ is the forest consisting of all subtrees  rooted at the apex children of $v_i$ in $T'$, for $i=1,\dots,s$.  For every $1\leq i<s$, we have $|F'_i|<|T'|/2$, because $v_i$ has a heavy child $v_{i+1}$ whose subtree contains at least $|T'|/2$ nodes. Thus, the spine of $T'^*$ contains the root $v_1$ of $T'$ and  the last heavy node $v_s$ on the spine $S$ (which can be~$v_1$ itself). Thus the spine of $T'^*$ contains at most two nodes. This completes the proof of the first part of the lemma. 
Assume now that $v$ is a strict ancestor of $u$ in $T^*$. If $v$ is an apex, then $v$ is an ancestor of $\apex(u)$ in $T^*$. If $v$ is heavy, then all its children are apexes, and therefore it follows that $v$ is an ancestor of $\apex(u)$ in $T^*$. 
\end{proof}

Observe that, as mentioned above, the ancestry relation in $T$ is partially preserved in $T^*$. To see why, notice first that no new ancestry relations are created in $T^*$, that is, if $v$ is an ancestor of $u$ in~$T^*$ then $v$ is also an ancestor of $u$ in $T$. Second, consider a node $v$ and one of its descendants $u$ in $T$. If $v$ is an apex node in $T^*$, then $v$ is also an ancestor of $u$ in $T^*$. If $v$ is heavy, namely, $v=v_i$,  for some $2\leq i\leq s$ in a spine $v_1,v_2,\cdots, v_s$, then there are two cases. If $u$ belongs to $F_i$ then $v$ is also an ancestor of $u$ in $T^*$.
However, if $u\in F_j$ for $j>i$, then the ancestry relation between $v$ and $u$  is not preserved in $T^*$. (Note, in this case $u$ is a descendant of $v_j$ in both $T$ and $T^*$ but $v_i$ and~$v_j$ are siblings in $T^*$ -- each being a child of the corresponding apex node $v_1$). 
In other words, the only case where the ancestry relation between a node $v$ and its descendant $u$ in $T$ is not preserved in~$T^*$ is when $v$  is heavy, i.e., $v=v_i$ in some spine $S=(v_1,v_2,\cdots, v_s)$ where $1<i< s$, and $u$ is a descendant in $T$ of $v$'s heavy child $v_{i+1}$ (including the case where $u=v_{i+1}$). In this case, $v$ is an ancestor of $u$ in $T$ but not in $T^*$.  

For a node $v$, let $\apex(v)$ be the closest ancestor of $v$ in $T$ that is an apex. Note that if $v$ is an apex then $\apex(v)=v$. Recall that every node is either heavy or apex. 
Consider a DFS traversal over $T$ that starts from the root and visits apex children first. For any node $u$, let $\dfs(u)$ be the DFS number of $u$ (where the DFS number of the root is~1). 

The following lemma provides a characterization of the ancestry relations in $T$ in terms of the structure of $T^*$.

\begin{lemma}\label{folding}
For any two different nodes $v$ and $u$, we have: $v$ is an ancestor of $u$ in $T$ if and only if at least one of the following two conditions hold
\begin{itemize}
\item
{\bf C1.} Node $v$ is an ancestor of $u$ in $T^*$;
\item
{\bf C2.} Node $\apex(v)$ is a strict ancestor of $u$ in $T^*$ and $\dfs(v)<\dfs(u)$.
\end{itemize}
\end{lemma}

\begin{proof}
Consider first the case that $v$ is an ancestor of $u$ in $T$. If $v$ is an apex node in $T^*$, then, by construction, $v$ is also an ancestor of $u$ in $T^*$, and thus Condition C1 holds. If $v$ is heavy, namely, $v=v_i$,  for some $2\leq i\leq s$ in some spine $S=(v_1,v_2,\cdots, v_s)$, then there are two cases. If $u$ belongs to $F_i$ then $v$ is also an ancestor of $u$ in $T^*$, and thus Condition C1 holds.
If $u\in F_j$ or $u=v_j$, for $j>i$, then we show that Condition C2 holds. First,  since $v$ is an ancestor of $u$ in $T$, we immediately get that $\dfs(v)<\dfs(u)$. 
Second, $\apex(v)$ is the apex $v_1$ which is, by definition, an ancestor of all nodes in the subtree of $T$ rooted at $v_1$. Therefore, since $\apex(v)$ is an apex, $\apex(v)$ is  
 an ancestor of $u$ in $T^*$. In fact,  $\apex(v)$ is  
 a strict ancestor of $u$ in $T^*$ since $u=v_j$ for $j>i >1$, and thus $u \neq v_1=\apex(v)$.

Conversely, consider the case that either Condition C1 or Condition C2 holds.
If Condition C1 holds, i.e., $v$ is an ancestor of $u$ in $T^*$, then the fact that $v$ is an ancestor of $u$ in $T$ is immediate since no  new ancestry relations are created in $T^*$. (Indeed, for any two nodes $w$ and $w'$, if $w$ is a parent of $w'$ in $T^*$ then $w$ is an ancestor or $w'$ in $T$).
  Consider now the case that Condition C2 holds. Assume, by way of contradiction, that $v$ is not an ancestor of $u$ in $T$.  Since  $\dfs(v)<\dfs(u)$ it follows that  $u$ is not an ancestor of $v$, and hence $v$ and $u$ have a least common ancestor $w$ in~$T$ which is neither $v$ or $u$. Since the DFS traversal visits apex nodes first, and since $v$ is visited before~$u$, we get that $v$ is a descendant of one of $w$'s apex children $z$ (because $w$ has at most one heavy child). Hence, $\apex(v)$ is either the apex $z$ or a strict descendant of it in $T$. In either case, $\apex(v)$ is  not an ancestor of $u$ in $T$. This contradicts Condition C2 since no  new ancestry relations are created in $T^*$.
\end{proof}

\subsection{The labeling scheme}

\subsubsection{Informal description}
An important property of the folding decomposition is that its spine decomposition depth is at most~2 (cf. Lemma~\ref{lem:depth2}). 
Essentially, our ancestry scheme is based on Lemma \ref{folding}. We first apply the scheme described in Section \ref{sec:bounded} on $T^*$, resulting in each node $u$ having a label $I(u)$ which is an interval encoded using $\log n+O(1)$ bits. Given the intervals $I(u)$ and $I(v)$ of two nodes $u$ and $v$, we can then detect whether or not Condition C1 holds. Indeed, given two nodes $u$ and $v$, if $I(u)\subseteq I(v)$ holds, then $v$ is an ancestor of $u$ in~$T^*$, and hence also in $T$. This takes care of all  ancestry relations preserved under the folding decomposition, and in particular the case where $v$ is an apex. To handle unpreserved ancestry relations, it is sufficient, by Lemma \ref{folding}, to check whether or not Condition C2 holds. For this purpose, we would like our decoder to reconstruct not only the interval $I(v)$, but also the interval $I(\apex(v))$. Indeed, having $I(\apex(v))$ and $I(u)$ already enables to detect whether or not  the first part of Condition C2 holds, namely, whether   $\apex(v)$ is a strict ancestor of $u$ in~$T^*$. 
To detect whether or not $\dfs(v)<\dfs(u)$ we shall actually use a specific implementation of our ancestry labelling scheme from Section~\ref{sec:bounded}, which will relate the interval relation $\prec$ (defined in Section \ref{sec:preliminaries}) to the DFS numbering in $T$. This implementation will guarantee that, for any heavy node $v$, $\dfs(v)<\dfs(u)$ if and only if $I(v)\prec I(u)$.

It is now left to explain how to obtain, for every node $u$, the interval $I(\apex(u))$, given the label of $u$ resulting from the scheme described in Section \ref{sec:bounded}, and using few additional bits (apart from the ones used to encode $I(u)$).
First, we use one additional  bit of information in the label of each node $u$, for indicating whether or not $u$ is an apex.
In the case $u$ is an apex, we already have $u=\apex(u)$, and hence, no additional information is required to reconstruct $I(\apex(u))$. For the case that $u$ is heavy, we  use additional $O(\log\log n)$ bits of information at the label of $u$. Specifically, in addition to its own interval $I(u)$, every node $u$ stores the level of the interval $I(\apex(u))$, consuming $\lceil \log k \rceil$ bits. Now, notice that $I(u)\subset I(\apex(u))$. Moreover, we will later prove that, in the scheme described in Section \ref{sec:bounded}, for every level $k$, the number of intervals $J$ used by the scheme on level $k$, such that $I(u)\subset J$, is at most $B_k^2$ (where $B_k$ is defined in Subsection \ref{section:marker}).
Thus, once we know level $k$ of $I(\apex(u))$, and the interval $I(u)$, additional 
$2\log B_k$ bits are sufficient to  reconstruct $I(\apex(u))$.   Note that $2\log B_k=2\log k +O(\log\log k)$  because $B_k=O(k\log^2 k)$. Since $k\leq \log n$, the total information per node (stored in its label) amounts to  $\log n + 3\log \log n +O(\log\log \log n)$ bits.

We are now ready to describe our ancestry labeling scheme formally. 

\subsubsection{The marker algorithm $\cM$}

Given a forest $F$, we describe the labels assigned to the nodes of each tree $T$ in $F$, so that to enable detecting ancestry relations. First, we apply the spine decomposition to $T$. Next, we consider a DFS traversal over $T$ that starts from the root, numbered~1, and visits apex children first. For any node $u$, let $\dfs(u)$ be the DFS number of $u$. (These DFS numbers will not be directly encoded in the labels, since doing that would require the consumption of too many bits; however, these numbers will be used by the marker to assign the labels). From the spine decomposition of $T$, we construct the folding decomposition $T^*$ of $T$, as described in Subsection~\ref{sec:folding}. Recall from Lemma~\ref{lem:depth2} that the spine decomposition depth of $T^*$ is at most~2. 

Next, we apply to $T^*$ the ancestry scheme defined in Section~\ref{sec:bounded}. More specifically, we perform a particular implementation of this ancestry scheme, as described now. Consider a level $k$ of the recursive procedure applied by the marker for assigning labels to the nodes in $T^*$. At this level, the marker is dealing with some subtree $T^*{'}$ of $T^*$ with size at most $2^k$. We revisit the assignment of the intervals and bins considered by the marker at this level, for $T^*{'}$. Note that the nodes in $T^*{'}$ are the nodes of a subtree $T'$ of $T$. In fact $T^*{'}$ is the folding decomposition of $T'$. Thus, in order to avoid cumbersome notations, in the remaining of this description, we replace $T'$ and $T'^*$ by $T$ and $T^*$, respectively. The reader should just keep in mind that we are dealing with the general case of a given tree of a given level $k$ in the recursion, and not necessarily  with the whole tree in the first level of this recursion. So, given $T$, and $T^*$, both of size at most $2^k$, we show how to implement our ancestry scheme in a special manner on $T^*$. 

Recall the DFS traversal over the whole tree, which assigns a DFS number $\dfs(u)$ to every node~$u$ in the tree. This DFS traversal  induces a DFS traversal on (the subtree) $T$ that starts from the root of $T$ and visits apex children first. 

Let $S=(v_1,\dots,v_s)$ be the spine of $T$, and $F_i$ be the forest consisting of all subtrees rooted at the apex children of $v_i$ in $T$, for $i=1,\dots,s$. Hence, $T^*$ consists of a tree rooted at $v_1$, where $v_1$ has $s-1$ children $v_2,\dots,v_s$, plus $v_1$'s apex children $u_1,\dots,u_t$ in $T$ (see  Figure~\ref{fig:folding}). We order these nodes~$u_i$ so that $\dfs(u_i) < \dfs(u_{i+1})$. Therefore, we get the following important ordering implied by the DFS traversal: 
\begin{equation}\label{hahaha}
\dfs(u_1)<\dfs(u_2)<\dots < \dfs(u_t) < \dfs(v_2) < \dfs(v_3) < \dots < \dfs(v_{s-1}) < \dfs(v_s).
\end{equation}
As mentioned in the proof of Lemma~\ref{lem:depth2}, in $T^*$, $v_s$ is the heavy child of $v_1$, and all the other children of $v_1$ are apexes. Therefore, the spine of $T^*$ is $(v_1^*,v_2^*)=(v_1,v_s)$. (Recall, it maybe the case that $v_1=v_s$, in which case the spine would consist of a single node, namely $v_1$). Moreover, the forest~$F^*_1$ consists of all trees in $T^*$ hanging down from $u_1,\dots,u_t$, as well as all trees  in $T^*$ hanging down  from $v_2,\dots,v_{s-1}$. Similarly, the forest $F^*_2$ (if it exists) consists of all trees in $T^*$ hanging down from the children of $v_s$ in $T$. 

In this level $k$ of the recursion of our ancestry scheme, we are given a bin $J$ of size $\left\lfloor c_k|T^*|\right\rfloor$, and we need to allocate intervals to the nodes in $T^*$ contained in $J$. For this purpose, we assign intervals to $v_1^*$ and $v_2^*$, and bins to the forests $F^*_1$ and $F^*_2$. These bins are $J_1$ and $J_2$, and, by Equation~\ref{eq:prec}, we have $J_1\prec J_2$.  Moreover, if $v_1\neq v_s$, then, by Facts~\ref{fact2} and~\ref{fact3}, we get that:
\begin{equation}\label{eq:xxx}
\mbox{for every $v\in F_1^*$ and every $u\in F_2^*\cup\{v_2^*\}$, we have $I(v)\prec I(u)$.}
\end{equation}

Furthermore, observe that all trees in $F_1^*$ are of size at most $2^{k-1}$. Therefore, the assignment of intervals to the nodes of $F_1^*$ are performed according to the simpler case described in the proof of Lemma~\ref{lem:main-bounded}. We order the trees $T_i$, $i=1,\dots,t+s-2$, in $F_1^*$ according to the DFS numbers of their roots. In particular, for $i=1,\dots,t$, the root of $T_i$ is $u_i$, and, for $i=1,\dots,s-2$, the root of~$T_{t+i}$ is $v_{i+1}$. By Fact~\ref{fact1}, we get that the interval assignment satisfies the following. 
\begin{equation}\label{eq:yyy}
\mbox{For every $v\in T_i$ and $u\in T_j$ where $j>i$, we have $I(v) \prec I(u)$.}
\end{equation}
We are now ready to state the crucial relationship between the ancestry relations in the whole tree and the assignment of intervals to nodes in its folding decomposition. Essentially, this lemma replaces the dependencies on the DFS numbering mentioned in the statement of Lemma~\ref{folding} by a dependency on the relative positions of the intervals according to the partial order $\prec$. This is crucial because the DFS numbers are not part of the labels. 

\begin{lemma}\label{lem:characterization}
For any two different nodes $v$ and $u$, we have: $v$ is an ancestor of $u$ in $T$ if and only if at least one of the following two conditions hold
\begin{itemize}
\item  Node $v$ is an ancestor of $u$ in $T^*$;
\item Node $\apex(v)$ is a strict ancestor of $u$ in $T^*$ and $I(v) \prec I(u)$.
\end{itemize}
\end{lemma}

\begin{proof}
Given Lemma~\ref{folding}, we only need to prove that if $v$ is not an ancestor of $u$ in $T^*$, but $\apex(v)$ is a strict ancestor of $u$ in $T^*$, then 
\[
\dfs(v) < \dfs(u) \iff I(v) \prec I(u).
\]
Let us consider the subtrees of $T$ and~$T^*$, rooted at $\apex(v)$. By slightly abusing notation, we reuse the same notation as before. Recall the spine $S=(v_1,\dots,v_s)$, whose apex is $v_1=\apex(v)$. The children of $\apex(v)$ in~$T^*$ are the apex nodes $u_1,\dots,u_t,v_2,\dots,v_{s-1}$, plus the heavy node $v_s$. Each apex child  is the root of a tree $T_i$ for some $i\in\{1,\dots,t+s-2\}$. All these trees belong to~$F_1^*$. All the trees in $F_2^*$ are rooted at children of $v_s$. 

Since $v$ is not an ancestor of $u$ in $T^*$, but $\apex(v)$ is an ancestor of $u$ in $T^*$, we get that $v\neq \apex(v)$. It follows that $v$ must be one of the spine nodes $v_2,\dots,v_s$, say $v=v_j$ with $j\geq 2$. Node $u$ is a strict descendent of $\apex(v)$, but is not a descendent of $v$. Therefore, $u$ belongs either to one of the trees $T_i$, $i=1,\dots,t+s-2$, $i\neq j$, in $F_1^*$, or to $F_2^*\cup\{v_s\}$ in case $v\neq v_s$. 

Assume first that $\dfs(v)<\dfs(u)$. In that case, $v$ cannot be $v_s$. Indeed, if $v=v_s$ then $u$ cannot be in $F^*_2\cup \{v_s\}$ because $v$ is not an ancestor of $u$. Hence, $u$ belongs to one of the trees~$T_i$. This contradicts the fact that $\dfs(v)<\dfs(u)$, by Equation~\ref{hahaha}. So, $v=v_j$, $2\leq j < s$. Since $\dfs(v)<\dfs(u)$, we know that  $u$ belongs either to one of the trees $T_i$, $i=t+j+1,\dots,t+s-2$, or to $F_2^*\cup\{v_s\}$. In the former case, we get $I(v)\prec I(u)$ by applying Equation~\ref{eq:yyy}; In the latter case, we get $I(v)\prec I(u)$ by applying Equation~\ref{eq:xxx}.

Conversely, assume that $I(v)\prec I(u)$. In that case, again $v$ cannot be $v_s$. Indeed, if $v=v_s$ then~$u$ cannot be in $F_2^*\cup\{v_s\}$ because $v$ is not a ancestor of $u$. Hence, $u$ belongs to one of the trees~$T_i$. This contradicts the fact that $I(v)\prec I(u)$ by Equation~\ref{eq:yyy}. So $v=v_j$, $2\leq j < s$. In that case, $u$ belongs  to one of the trees $T_i$, $i=t+j+1,\dots,t+s-2$, of $F_1^*$ or to $F^*_2\cup \{v_s\}$, which implies that $\dfs(v)<\dfs(u)$ by Equation~\ref{hahaha}. 
\end{proof}

\begin{figure}
\begin{center}
\includegraphics[width=0.9\linewidth]{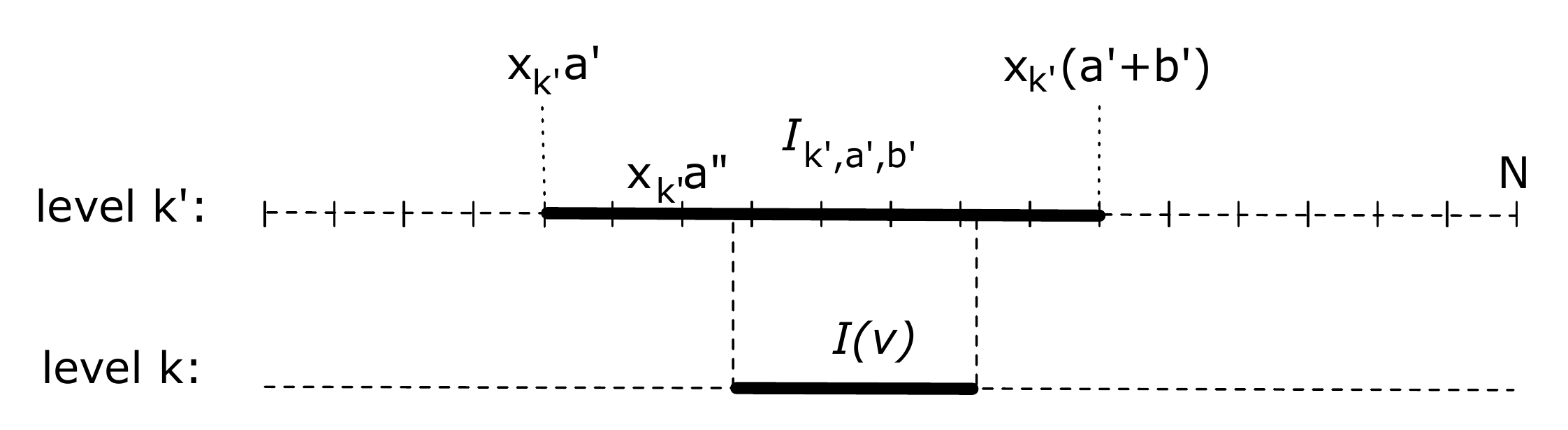}
\caption{Encoding $I(\apex(v))$ using $I(v)$ and few additional bits}
\label{fig:inclusion}
\end{center}
\end{figure}

Lemma~\ref{lem:characterization} provides a characterization for ancestry in $T$ with respect to properties of the intervals assigned to the nodes of $T^*$. More  specifically, the fact that $v$ is an ancestor of $u$ can be detected using the intervals $I(v)$ and $I(u)$, as well as the interval $I(\apex(v))$ of $\apex(v)$. The interval~$I(v)$ is encoded directly in the label of $v$ using $\log n+O(1)$ bits. Directly encoding $I(\apex(v))$ would consume yet another $\log n+O(1)$ bits, which is obviously undesired. We now show how to encode~$I(\apex(v))$ in the label of $v$ using few  bits of information in addition to the ones used to encode~$I(v)$. We use the following trick (see Figure~\ref{fig:inclusion}).
Let $k',a'$, and $b'$ be such that
$$I(\apex(v))=I_{k',a',b'}=[x_{k'}a',x_{k'}(a'+b')].$$ Since $1\leq k \leq \log n$ and $1\leq b'\leq B_{k'}$, we get that $2\log\log n +O(\log\log\log n)$ bits suffice to encode both  $k'$ and $b'$.
To encode $a'$, the marker algorithm acts as follows. Let $I(v)=[\alpha,\beta]$, and let $a''$ be the largest integer such that $x_{k'} a''\leq \alpha$.
 We have, $a''- B_{k'}\leq a''-b'$ because $b'\leq B_{k'}$. Since $I(v)\subseteq I(\apex(v))$, we also have $x_{k'}(a'+b')\geq \beta \geq \alpha \geq x_{k'} a''$. Thus $a''-b'\leq a'$. Finally, again since $I(v)\subseteq I(\apex(v))$, we have $x_{k'}a'\leq \alpha$, and thus $a''\geq a'$. Combining the above inequalities, we get that $a'\in[a''-B_{k'},a'']$. The marker algorithm now encodes the integer $t\in [0,B_{k'}]$ such that $a'=a''-t$. This is done in consuming another $\log B_{k'}=\log\log n+O(\log\log\log n)$ bits.
Hence, we obtain:

\begin{lemma}\label{lemma:labelsize}
The marker algorithm $\cM$ assigns labels to the nodes of $n$-node trees in linear time,
and each label is encoded using $\log n+3\log\log n+O(\log\log\log n)$ bits.
\end{lemma}

\begin{proof}
Let $v$ be a node of an $n$-node tree $T$. The label of $v$ consists of the interval $I(v)$ plus the additional information that will later be used to extract $I(\apex(v))$, namely, according to the notations above, the values of $t$, $k'$, and $b'$. The leftmost part of the label $L(v)$ will be dedicated to encode these latter values $t$, $k'$, and $b'$. Each of these values can be separately encoded using $\log\log n+O(\log\log\log n)$ bits. By adding $O(\log\log\log n)$ bits per value, we can describe where the binary encoding of each of these values starts and ends. Therefore, in total, encoding these three values together with their positioning in the label requires $3\log\log n+O(\log\log\log n)$ bits. Now, since $I(v)$ is encoded using $\log n+O(1)$ bits, each label uses $\log n+3\log\log n+O(\log\log\log n)$ bits. 
\end{proof}

Note that the bound on the label size in Lemma~\ref{lemma:labelsize} includes the encoding of the positioning  of the different fields in the label, each of which being encoded using $O(\log\log\log n)$ bits. 

\subsection{The decoder algorithm $\cD$}\label{sub:decoder}

Now, we describe our decoder algorithm $\cD$. Given the labels $L(v)$ and $L(u)$ assigned by $\cM$ to two different nodes $v$ and $u$ in some tree $T$, the decoder algorithm $\cD$ needs to find whether $v$ is an ancestor of $u$ in $T$. (Observe that since each node receives a distinct label, the decoder algorithm can easily find out
if $u$ and $v$ are in fact the same node, and, in this trivial case, it simply outputs~1.)

The decoder algorithm first extracts the three values $t$, $k'$, and $b'$ by inspecting the first  $3\log\log n+O(\log\log\log n)$ bits of $L(v)$,  and then the interval $I(v)$ by inspecting the remaining  $\log n+O(1)$ bits of $L(v)$.  Using this information, it computes $I(\apex(v))$. Now, given the two intervals $I(v)$ and $I(\apex(v))$, as well as the interval $I(u)$ of a node $u$, the decoder acts according to the characterization stated in Lemma~\ref{lem:characterization}, that is:
$D(L(v),L(u))=1$ if and only if at least one of the following two conditions holds:
\begin{itemize}
\item $I(u)\subseteq I(v)$;
\item $I(u) \subset I(\apex(v))$ and $I(v) \prec I(u)$.
\end{itemize}
The fact that $\cD(L(v),L(u))=1$ if and only if  $v$ is an ancestor of $u$ in $T$ follows from Lemma~\ref{lem:characterization} and the correctness of the interval assignment to the nodes of $T^*$ given by Lemma~\ref{lem:correct}. This establishes the following result. 

\begin{lemma}\label{lem:correct2}
$(\cM,\cD)$ is a correct ancestry-labeling scheme for the family of all forests. 
\end{lemma}

By combining Lemmas \ref{lemma:labelsize} and \ref{lem:correct2}, we obtain  
 the main result of the paper.

\begin{theorem}\label{theo:main} 
There exists an ancestry-labeling scheme for the family of all forests, using label size $\log n+3\log\log n+O(\log\log\log n)$  bits for $n$-node forests. Moreover, given an $n$-node forest $F$, the labels can be assigned to the nodes of $F$ in $O(n)$ time,
and any ancestry query can be answered in constant time.
\end{theorem}

\paragraph{Remark:}
The ancestry-labeling scheme described in this section uses labels of optimal size $\log n+O(\log\log n)$ bits, to the price of a decoding mechanism based of an interval condition slightly more complex than the simple interval containment condition. Although this has no impact on the decoding time (our decoder still works in constant time), the question of whether there exists an ancestry-labeling scheme with labels of size $\log n+O(\log\log n)$ bits, but using solely the  interval containment condition, is intriguing. In Section \ref{sec:bounded}, we have shown that, at least in the case of trees with bounded spine-decomposition depth, extremely small labels can indeed  be coupled with an interval containment decoder.

\section{Small universal posets}\label{sec:poset}

%
%

Recall that an ancestry-labeling scheme  for $\cF(n)$ is \emph{consistent} if  the decoder satisfies the anti-symmetry and transitivity conditions (see Section~\ref{subsec:consistent}).
Our next lemma relates compact consistent ancestry-labeling schemes with small universal posets.

\begin{lemma}\label{labeling-poset}
There exists a consistent ancestry-labeling scheme for $\cF(n)$ with  label size $\ell$ if and only if for every integer $k$,
there exists a universal poset of size $2^{k\ell}$ for $\cP(n,k)$.
\end{lemma}

\begin{proof}
Assume first that for every integer $k$ there exists a universal poset $(U,\preceq)$ of size $2^{k\ell}$ for $\cP(n,k)$.
In particular, there exists a universal poset $(U,\preceq)$ of size $2^{\ell}$ for $\cP(n,1)$, i.e., for the family of $n$-element posets with tree-dimension $1$.
We construct the following ancestry-labeling scheme $(\cM,\cD)$ for $\cF(n)$.
The marker algorithm  $\cM$ first considers some bijective mapping $\phi:U\to[1,|U|]$.
Given a rooted forest $F\in\cF(n)$, view $F$ as a poset  whose Hasse diagram is $F$. I.e., $F$ is a poset with tree-dimension 1. Now, let $\rho:F\to U$ be a mapping that preserves the partial order of $F$.
The marker algorithm assigns the label $$L(u)=\phi(\rho(u))$$ to each node $u\in F$. Given the labels $L(u)$ and $L(v)$
of two nodes $u$ and $v$ in some $F\in \cF(n)$, the decoder algorithm acts as follows:  $$\cD(L(u),L(v))=1 \iff \phi^{-1}(L(u))\leq \phi^{-1}(L(v))~. $$
By construction,  $(\cM,\cD)$ is a consistent ancestry-labeling scheme for $\cF(n)$ with label size $\log|U|=\ell$.

For the other direction, assume that there exists a consistent ancestry-labeling scheme $(\cM,\cD)$ for $\cF(n)$ with label size $\ell$. Note that it may be the case that $ \cD(a,b)$ is not defined for some $a,b\in [1,2^{\ell}]$, which may happen if the marker $\cM$ is not using all values in $[1,2^{\ell}]$. In that case, we set $ \cD(a,b)=\bot$, that is $\cD$ has now three possible outputs $\{0,1,\bot\}$.
Let $U=[1,2^{\ell}]^k$. We define a relation $\preceq$ between pairs of elements in $U$.
For two elements $u,v\in U$, where $u=(u_1,u_2,\cdots, u_k)$ and $v=(v_1,v_2,\cdots, v_k)$, we set $$u \preceq v \iff 
\forall i\in [1,k],  \; \cD(u_i,v_i)  = 1~.$$
Recall that every poset in $\cP(n,k)$ is the intersection of $k$ tree-posets, each of which having a Hasse diagram being a forest. Since the ancestry relations in those forests precisely captures the ordering in the forest, it follows that $(U,\preceq)$ is a universal poset  for $\cP(n,k)$.
The lemma follows.
\end{proof}

For the purpose of constructing a small universal poset for $\cP(n,k)$, let us revise the ancestry-labeling scheme $(\cM,\cD)$ for $\cF(n)$ given by Theorem~\ref{theo:main}, in order to make it consistent.
We  define a new decoder algorithm~$\cD'$ as follows. Recall that the marker $\cM$ assigns to each node $u$ a label $L(u)$ that encodes the pair $(I(u),I(\apex(u))$. Given two labels $L(v)$ and $L(u)$, we set the modified decoder as follows: $\cD'(L(v),L(u))=1$ if and only if
at least one of the following three conditions holds:
\begin{itemize}
\item \textbf{[D0]:} $L(u) = L(v)$;
\item \textbf{[D1]:}  $I(u)\subset I(v)$ and $I(\apex(u))\subseteq I(v)$;
\item \textbf{[D2]:} $I(u) \subset I(\apex(v))$, $I(v) \prec I(u)$, and $I(\apex(u))\subseteq I(\apex(v))$.
\end{itemize}

\begin{lemma}\label{lemma:consistent}
$(\cM,\cD')$ is a consistent ancestry-labeling scheme for $\cF(n)$ whose label size is $\log n +3\log\log n +O(\log\log\log n)$ bits.
\end{lemma}

\begin{proof}
The required label
size of $(\cM,\cD')$ follows directly from Theorem~\ref{theo:main}.
The correctness of the ancestry-labeling scheme $(\cM,\cD')$ follows from the observation that the additional conditions to the definition of the decoder $\cD$ are redundant in the case $v$ and  $u$ belong to the same forest. 
To see why this observation holds, we just need to consider the situation in which $v$ in an ancestor of  $u$ in some tree, and check that $\cD'(L(v),L(u))=1$. By the correctness of~$\cD$, since $v$ is an ancestor of $u$, we get that either
\begin{itemize} 
\item  $I(u)\subseteq I(v)$, or
\item  $I(u) \subset I(\apex(v))$ and $I(v) \prec I(u)$. 
\end{itemize} 

-- If $I(u)\subseteq I(v)$, then either $I(u) = I(v)$ or $I(u)\subset I(v)$. In the former case, we actually have $u=v$ because $u$ and $v$ are in the same tree, and thus \textbf{D0} holds. In this latter case, $u$ is a strict descendant of $v$ in the folding decomposition $T^*$.  By Lemma~\ref{lem:depth2}, $\apex(u)$ is a descendent of $v$ in $T^*$, implying that \textbf{D1} holds. 

-- If $I(u) \subset I(\apex(v))$ and $I(v) \prec I(u)$, then just the strict inclusion  $I(u) \subset I(\apex(v))$ already implies  $I(\apex(u)) \subseteq I(\apex(v))$, and hence \textbf{D2} holds. 

It remains to show that $(\cM,\cD')$ is consistent. To establish the anti-symmetry property, let $\ell_1$ and $\ell_2$ be two different labels in the domain of $\cM$, and assume that $\cD'(\ell_1,\ell_2)=1$. We show that  $\cD'(\ell_2,\ell_1)=0$. Let $u$ and $v$ be two nodes, not necessarily in the same forest, such that $L(v)=\ell_1$ and $L(u)=\ell_2$. Since $\ell_1\neq \ell_2$, we have that either \textbf{D1} or \textbf{D2} holds. Therefore, either $I(u)\subset I(v)$ or $I(v)\prec I(u)$. This implies that $\cD'(L(u),L(v))=0$, by the anti-symmetry of both relations $\subset$ and $\prec$. 

To establish the transitivity property, let $\ell_1$, $\ell_2$ and $\ell_3$ be three pairwise different labels in the domain of $\cM$, and assume that $\cD'(\ell_1,\ell_2)=1$ and $\cD'(\ell_2,\ell_3)=1$. We show that  $\cD'(\ell_1,\ell_3)=1$. Let $u$, $v$, and $w$ be three nodes, not necessarily in the same forest, such that $L(w)=\ell_1$, $L(v)=\ell_2$, and $L(u)=\ell_3$. Since the three labels are pairwise different,  \textbf{D0} does not hold, and thus we concentrate the discussion on \textbf{D1} and \textbf{D2}. In other words, we must have the situation: 
\begin{itemize}
\item \textbf{[D1$(v,u)$]:}  $I(u)\subset I(v)$ and $I(\apex(u))\subseteq I(v)$; or 
\item \textbf{[D2$(v,u)$]:} $I(u) \subset I(\apex(v))$, $I(v) \prec I(u)$, and $I(\apex(u))\subseteq I(\apex(v))$.
\end{itemize}
and
\begin{itemize}
\item \textbf{[D1$(w,v)$]:}  $I(v)\subset I(w)$ and $I(\apex(v))\subseteq I(w)$; or
\item \textbf{[D2$(w,v)$]:} $I(v) \subset I(\apex(w))$, $I(w) \prec I(v)$, and $I(\apex(v))\subseteq I(\apex(w))$.
\end{itemize}
We examine each of the four combinations of the above conditions, and show that, for each of them, at least one of the following two conditions holds: 
\begin{itemize}
\item \textbf{[D1$(w,u)$]:}  $I(u)\subset I(w)$ and $I(\apex(u))\subseteq I(w)$; 
\item \textbf{[D2$(w,u)$]:} $I(u) \subset I(\apex(w))$, $I(w) \prec I(u)$, and $I(\apex(u))\subseteq I(\apex(w))$.
\end{itemize}

\noindent -- \textbf{Case 1.1:} \textbf{[D1$(v,u)$]} and \textbf{[D1$(w,v)$]} hold. We get immediately that $I(u)\subset I(w)$, and $I(\apex(u))\subseteq I(w)$ by transitivity of $\subseteq$ and $\subset$, and thus \textbf{[D1$(w,u)$]} holds. 

\noindent -- \textbf{Case 1.2:} \textbf{[D1$(v,u)$]} and \textbf{[D2$(w,v)$]} hold. We show that \textbf{[D2$(w,u)$]} holds. First, we have  $I(u)\subset I(v) \subset I(\apex(w))$. Second, $I(u)\subset I(v)$  and $I(w) \prec I(v)$, so  $I(w) \prec I(u)$. Finally, $I(\apex(u))\subseteq I(v)\subseteq I(\apex(v))\subseteq I(\apex(w))$. Thus \textbf{[D2$(w,u)$]} holds. 

\noindent -- \textbf{Case 2.1:} \textbf{[D2$(v,u)$]} and \textbf{[D1$(w,v)$]} hold. We show that \textbf{[D1$(w,u)$]} holds. First, $I(u) \subset I(\apex(v)) )\subseteq I(w)$. Second, $I(\apex(u))\subseteq I(v))\subseteq I(w)$. Thus \textbf{[D1$(w,u)$]} holds.

\noindent -- \textbf{Case 2.2:} \textbf{[D2$(v,u)$]} and \textbf{[D2$(w,v)$]} hold. We show that \textbf{[D2$(w,u)$]} holds. First, $I(u) \subset I(\apex(v))\subseteq I(\apex(w))$. Second, $I(w) \prec I(v) \prec I(u)$. Finally, $I(\apex(u))\subseteq I(\apex(v))\subseteq I(\apex(w))$. Thus \textbf{[D2$(w,u)$]} holds.

This completes the proof of the lemma.
\end{proof}

By combining Lemma~\ref{labeling-poset} and Lemma~\ref{lemma:consistent}, we obtain the following.

\begin{theorem}\label{theorem:poset}
For every integer $k$, there exists a universal poset of size $O(n^k\log^{3k+o(k)} n)~$ for $~\cP(n,k)$.
\end{theorem}

\section{Further work}\label{sec:conclusion} 

In this paper, we solved the ancestry-labeling scheme problem.
In general, by now, the area of informative labeling-schemes is quite well understood in the deterministic setting. In particular, for many of the classical problems, tight bounds for the label size are established (see Section~\ref{sec:related}).  
In contrast, the randomized framework, initiated in \cite{FK-spaa09}, remains widely open. We conclude this paper by mentioning one open problem in the framework of randomized ancestry-labeling schemes in order to illustrate the potential of randomization in this domain.

We describe a simple one-sided error
ancestry-labeling scheme with label size $\lceil \log n\rceil$ bits.  The scheme guarantees that if $u$ is an ancestor of $v$, then, given the corresponding labels $L_u$ and $L_v$, the decoder is correct, that is, $\cD(L_u,L_v)=1$,  with probability 1; and if $u$ is not an ancestor of $v$ then the decoder is correct, that is, $\cD(L_u,L_v)=0$, with probability at least $1/2$. 
The scheme operates as follows.
Given a tree $T$, the marker first randomly chooses, for  every node $v$, an ordering of $v$'s children in a uniform manner, i.e., every ordering of the children of $v$ has the same probability to be selected by the marker. Then, according to the resulted orderings, the marker performs a  DFS traversal that starts at the root $r$, and labels each node it visits by its DFS number. Given two labels $i,j$, the decoder outputs:
\[
\cD(i,j)=\left\{
\begin{array}{ll}
1 & \mbox{if $i<j$;}\\
0 & \mbox{otherwise. }
\end{array}
\right.
\]
Let $u$ and $v$ be two nodes in $T$, and let $L_u=\dfs(u)$ and $L_v=\dfs(v)$ denote their labels. If $u$ is an ancestor of $v$, then $L_u<L_v$, no matter which orderings were chosen by the marker. Thus, in this case, $\cD(L_u,L_v)=1$ with probability~1. Now,
if $u$ is not an ancestor of $v$, we consider two cases. First, if $u$ is a descendant of $v$ then $L_u>L_v$ and therefore $\cD(L_u,L_v)=0$. If, however, $u$ and $v$ are non-comparable, i.e., neither one is an ancestor of the other, then the probability that the DFS tour visited $u$ before it visited $v$ is precisely $1/2$, i.e., $\Pr[L_u<L_v]=1/2$. Hence the guarantee for correctness in this case is $1/2$. 
It was also shown in \cite{FK-spaa09} that a one-sided error ancestry-labeling scheme with constant  guarantee  must use labels of size $\log n - O(1)$. 
An interesting open question  is whether for every constant $p$, where $1/2\leq p<1$, there exists a one-sided error ancestry-labeling scheme with guarantee $p$ that uses labels of size $\log n +O(1)$.




\bigskip
\noindent{\bf Acknowledgments.} The authors are very thankful to William T. Trotter, Sundar Vishwanathan  and Jean-Sebastien S\'ereni for helpful discussions.



\end{document}